\documentclass[a4paper]{scrartcl}
\usepackage{amssymb}
\usepackage{amsmath}
\usepackage{amsthm}
\usepackage{mathabx}
\usepackage{latexsym}
\usepackage{mathtools}
\usepackage{enumerate}
\usepackage{graphicx,color}
\usepackage{hyperref}

\usepackage{dsfont}

\usepackage{xcolor}
\usepackage{enumitem}

\usepackage[backend=bibtex,bibencoding=ascii,style=authoryear-comp,natbib=true,dashed=false,maxcitenames=2]{biblatex} 

\renewbibmacro{in:}{%
  \ifentrytype{article}{}{\printtext{\bibstring{in}\intitlepunct}}}
\DeclareFieldFormat[article]{title}{#1}
\DeclareFieldFormat[incollection]{title}{#1}
\DeclareFieldFormat[inbook]{title}{#1}
\DeclareFieldFormat[article]{title}{#1}
\DeclareFieldFormat[unpublished,thesis]{title}{\textit{#1}}
\DeclareFieldFormat[article]{volume}{\textbf{#1}}
    \DeclareFieldFormat[article]{number}{\mkbibparens{#1}}
    \renewbibmacro*{volume+number+eid}{
      \printfield{volume}
      \setunit*{\unspace}
      \printfield{number}
      \setunit*{\unspace}
      \printfield{eid}
      }
\DeclareFieldFormat{journaltitle}{\mkbibemph{#1}\isdot}

\DeclareCiteCommand{\citeyear}
  {\boolfalse{citetracker}%
   \boolfalse{pagetracker}%
   \usebibmacro{prenote}}
  {\printtext[bibhyperref]{\printfield{year}}}
  {\multicitedelim}
  {\usebibmacro{postnote}}

\usepackage[autostyle=true]{csquotes} 

\renewcommand{\d}{\mathrm{d}}
\newcommand{\E}{\mathbb{E}}

\DeclareMathOperator*{\plim}{\mathrm{plim}}

\makeatletter
\newcommand*\bigcdot{\mathpalette\bigcdot@{.5}}
\newcommand*\bigcdot@[2]{\mathbin{\vcenter{\hbox{\scalebox{#2}{$\m@th#1\bullet$}}}}}
\makeatother

\numberwithin{equation}{section}

\numberwithin{figure}{section}

\theoremstyle{plain}
\newtheorem{theorem}{Theorem}[section]
\newtheorem{proposition}[theorem]{Proposition}
\newtheorem{corollary}[theorem]{Corollary}
\newtheorem{lemma}[theorem]{Lemma}

\theoremstyle{definition}
\newtheorem{definition}[theorem]{Definition}

\newtheorem{example}[theorem]{Example}

\usepackage{authblk}

\title{A General Surplus Decomposition Principle in Life Insurance}

\author[1,$\star$]{Julian Jetses}
\author[1]{Marcus C.~Christiansen}

\affil[1]{\footnotesize Institut f{\"u}r Mathematik, Carl von Ossietzky Universit{\"a}t Oldenburg,  Carl-von-Ossietzky-Stra{\ss}e 9--11, DE-26129 Oldenburg, Germany.}

\affil[$\star$]{\footnotesize Corresponding author. E-mail: \href{mailto:Julian.Jetses@uni-oldenburg.de}{Julian.Jetses@uni-oldenburg.de}.}

\date{\today}
\allowdisplaybreaks

\begin{document}

\maketitle

\begin{abstract}
In with-profit life insurance,  the prudent valuation of future insurance liabilities leads to systematic surplus that mainly belongs to the policyholders and is redistributed as bonus. For a fair and lawful redistribution of surplus the insurer needs to decompose the total portfolio surplus with respect to  the contributions of individual policies and with respect to different risk sources. For this task, actuaries have a number of heuristic decomposition formulas, but an overarching decomposition principle is still missing. This paper fills that gap by introducing a so-called ISU decomposition principle that bases on infinitesimal sequential updates of the insurer's valuation basis. It is shown that the existing heuristic decomposition formulas can be replicated as ISU decompositions.  Furthermore, alternative decomposition principles and their relation to the ISU decomposition principle are discussed.
The generality of the ISU concept makes it a useful tool also beyond classical surplus decompositions in life insurance.
\end{abstract}


\section{Introduction}

The long term nature of life insurance contracts causes a significant trend risk that is non-diversifiable in the insurance portfolio. In with-profit life insurance the insurer uses a conservative trend scenario upfront and successively replaces it with the empirically observed trend. This updating of the valuation basis produces surplus.
Insurance regulation requires that this surplus is for the most part redistributed to the policyholders.
In each reporting period the insurer  determines at first the total profits and losses in the life insurance portfolio, which includes surplus from non-diversifiable as well as diversifiable risks.
In a second step the profits and losses are shared between the insurer and each of the policyholders subject to regulation. Surplus from diversifiable risk is commonly credited or debited to the insurer.
Surplus that arises from conservative trend scenario calculations mainly belongs to the policyholders.
In Germany the regulator moreover requires to distinguish between surplus due to demographic trends,  investment success  and  changes of administration costs. For such  splitting of the total surplus the insurer needs an additive and risks-based  decomposition method.

The life insurance literature knows various surplus decomposition formulas, cf.~Ramlau-Hansen (1988, 1991) and Norberg (1999) for continuous time modelling and Milbrodt \& Helbig (1999) for discrete time modelling.  All these formulas are derived in a heuristic manner. An overarching general decomposition principle is still missing, and this paper closes that gap. We introduce a so-called infinitesimal sequential updating (ISU) decomposition principle, which  can reproduce the existing decomposition formulas and puts them into a general and consistent framework. The ISU concept is an advancement of  sequential updating (SU) decomposition principles, which  are used in various fields of economics but have the disadvantage that they depend on a formal ordering of the different surplus sources, cf.~Fortin et al.~(2011) and Biewen (2014). The ISU concept overcomes this drawback of the SU concept by pushing the lengths of the reporting periods down to zero so that  the impact of the ordering vanishes. To our knowledge, this asymptotic approach is completely new in the literature.
An alternative to sequential decompositions are one-at-a-time (OAT) principles, which avoid the ordering problem but suffer from interaction effects between the parameters, undermining the desired additivity of surplus decompositions. We show that the asymptotic approach can help also here and find that  the resulting infinitesimal OAT decompositions are largely equivalent to ISU decompositions.

A recent trend in  insurance is to reward risk averse behaviour of the insured by means of individual activity tracking. The difference between the expected activity and the real activity of an insured leads to surplus, and advanced decomposition formulas are needed that can separate this activity surplus from the classical surplus sources. The ISU decomposition principle offers the necessary basic tools for solving that  problem, but a detailed study of these new insurance forms is beyond the scope of this paper and is left for future research.

The paper is structured as follows.
In section 2 we formally define the surplus decomposition problem. Section 3 describes the life insurance modelling framework and recalls the definition of the total surplus. In section 4 we introduce the ISU decomposition principle. Section 5 is rather technical and develops integral representations results for ISU decompositions, which are needed and applied in section 6, where we illustrate the ISU decomposition concept for typical life insurance applications.
 Section 7 discusses alternatives to the ISU decomposition concept and explains their relations. Section 8 briefly summarizes our findings.

\section{The surplus process of an individual insurance contract}

We generally assume that we have a complete probability space $(\Omega , \mathcal{A},\mathbb{P})$   with a right-continuous and complete filtration $\mathbb{F}=(\mathcal{F}_t)_{t \geq 0}$. We consider an individual insurance policy on a finite contract period $[0,T]$. For each $t \geq 0$ let $B(t) $ be the aggregated insurance cash flow  on $[0,t]$ between insurer and insured. We use the convention that premiums have a negative sign and benefits have a positive sign. Let  $\kappa$ be a semimartingale with $\kappa(0)=1$ that describes the value process of the insurer's self-financin investment portfolio. Then the value $A(t)$ of the assets accrued at time $t$ is given by
\begin{align}\label{DefA}
	A(t) =    - \int_{[0,t]} \frac{\kappa(t)}{\kappa(s)}  \d B(s),
\end{align}
assuming that $B$ is a finite variation semimartingale and that $\kappa$ is strictly positive.
In the hypothetical case that the  insurer knew the future, the liabilities at time $t$ would be likewise calculated as
\begin{align*}
	L^h(t) =  \int_{(t,T]} \frac{\kappa(t)}{\kappa(s)}  \d B(s).
\end{align*}
The difference between assets and liabilities is the surplus,
\begin{align}\label{ALformula}
	S^h(t)=  A(t) - L^h(t) = \kappa(t) ( A(0) - L^h(0)).
\end{align}
In this hypothetical setting,  the actual surplus emerges at time zero and any dynamics after zero just comes from the compounding factor $\kappa(t)$.  By defining
\begin{align*}
	\d \Phi(t) = \frac{\d \kappa(t)}{\kappa(t-)}
\end{align*}
as the return on investment of the insurers investment portfolio, the process $S^h$ satisfies
\begin{align*}
	\d S^h(t) = S^h(t-) \d \Phi(t)
\end{align*}
for  $t>0$, which shows again that the dynamics of $S^h$ on $(0,\infty)$ stems solely  from investment gains earned on the existing surplus. Since  $A(0)-L(0)$  depends on the future and is nowhere adapted to the available information,  in real life the insurer has to replace $A(0)-L(0)$  at each time $t$ by an $\mathcal{F}_t$-measurable proxy $R(t)$.  We denote $$R=(R_t)_{t \geq 0}$$ as the \emph{revaluation surplus process}  since it describes profits and losses that result from the continuous revaluation of $A(0)-L(0)$ as the information $\mathcal{F}_t$ increases with time $t$.
Now the \emph{total surplus process} is given by
	\begin{align}\label{totalSurplus}
  S(t) = \kappa(t) R(t), \quad t \geq 0,
\end{align}
and its dynamics is driven by both, the compounding factor $\kappa$ and the revaluation surplus process $R$.
The aim of this paper is to decompose $R$ with respect to a given set of risk sources.
We assume that  the life insurance model rests on a so-called \emph{risk basis}
$$X=(X_1 , \ldots, X_m),$$
  which is a multivariate adapted process composed of so-called \emph{risk factors} $X_1, \ldots, X_m$ such that $R$ is adapted to the right-continuous and  complete filtration generated by $X$. The information provided by $X$ at time $t$ can be represented by the stopped process $X^t$, formally defined by
\begin{align*}
X^t(s) = \mathds{1}_{s \leq t}\, X(s) + \mathds{1}_{s > t}\, X(t).
\end{align*}
So at each time $t $ the proxy $R(t)$ of $A(0)-L(0)$ can be interpreted as the value of a mapping
\begin{align*}
 (t,X^t) \mapsto R(t)
\end{align*}
that assigns at each time $t$ to the current information $X^t$ the  random variable $R(t)$.
In this paper we assume that there even exists a mapping $\varrho$ such that
\begin{align*}
 \varrho (X^t) =  R(t), \quad  t \geq 0.
\end{align*}
In the latter equation, the time parameter $t$ itself is not an argument of $\varrho$ and only appears as stopping parameter in $X^t$. That means that the dynamics of $R$ is solely driven by the  increase of information through $X^t$.

The central aim of this paper  is to decompose $R$ as
\begin{align}\label{aim:decomp}
  R(t) = R(0)+ D_1(t) + \cdots + D_m(t), \quad t \geq 0,
\end{align}
where $D_1, \ldots, D_m$ are adapted processes that start at zero and describe the contributions of each risk factor $X_1, \ldots, X_m$ to the dynamics of $R$. The first addend $R(0)$ represents initial surplus, which is not decomposed here.  Equation \eqref{aim:decomp} is equivalent to the additive decomposition
\begin{align}\label{totalSurplus2}
  S(t) = \kappa(t) S(0)+ \kappa(t) D_1(t) + \cdots + \kappa(t)D_m(t), \quad t \geq 0,
\end{align}
for the total surplus process. The first addend $\kappa(t) S(0)$ represents the time-$t$ value of the initial surplus $S(0)= R(0)$,  and the addends $\kappa(t) D_1(t), \ldots, \kappa(t) D_m(t)$ describe the time-$t$ values of the  contributions that the risk factors $X_1, \ldots, X_m$ make to the dynamics of $S$.
The additivity of the decompositions  \eqref{aim:decomp} and \eqref{totalSurplus2}
allows us to distribute the surplus among different parties.

The dynamics of the total surplus in \eqref{totalSurplus} is driven by investment gains on the surplus itself and by revaluation gains. In \eqref{totalSurplus2} the investment gains are subdivided among the different surplus contribution addends according to their shares in the total investment earnings.  It is not uncommon in the actuarial literature to collect all the investment gains in a separate term, see for example Norberg (1999, formula (5.3)). The idea is to apply It\^{o}'s product rule on $S(t)=\kappa(t) R(t)$ and then to identify each of the resulting addends either as investment gains or as revaluation gains. However, this approach mixes up the investment earnings of the carefully separated surplus contribution addends, so it is not helpful in our opinion and therefore it is not further considered in this paper.

\section{The revaluation surplus in multi-state models}\label{Section:LifeInsuranceModel}

Let the random pattern $Z$ of the insured be a right-continuous and adapted jump process on a finite state space $\mathcal{Z}$ with starting value $Z_0=a \in \mathcal{Z}$.
We define corresponding state processes $(I_j)_j$  and
counting processes  $(N_{jk})_{jk:j\neq k}$ by  $I_i(t)\coloneqq \mathds{1}_{\{Z(t)=i\}}$ and
\begin{align*}
     N_{jk}(t) = \sharp \{ s \in (0,t]: Z(s-)=j, Z(s)=k\}.
 \end{align*}
 Additionally, we define $N_{jj}\coloneqq -\sum_{k:k\neq j} N_{jk}$ and the vector-valued  process $N=(N_{jk})_{jk:j \neq k}$.

We call a pair $(\overline{\Phi}, \overline{\Lambda})$ a valuation basis if the following properties hold:
\begin{itemize}
\item  $\overline{\Phi}$ is semimartingale with $\overline{\Phi}(0)=0$ and $\Delta \overline{\Phi}(t)>-1$ for all $t> 0$,
\item $\overline{\Lambda}=(\overline{\Lambda}_{jk})_{jk: j\neq k}$ is a  vector-valued, right-continuous finite variation process with $\overline{\Lambda}(0)=0$,
\item the processes $\overline{\Lambda}_{jk}$, $j\neq k$ are non-decreasing and $\sum_{k:k \neq j} \Delta  \overline{\Lambda}_{jk}(t)\leq 1$ for every $t>0$ and every $j$.
\end{itemize}

The process $\overline{\Phi}$ represents cumulative returns on investment, and the solution
$\overline{\kappa}=(\overline{\kappa}(t))_{t\geq 0}$   of the stochastic differential equation
\begin{align}\label{SDEPhi}
&\d \overline{\kappa}(t)=\overline{\kappa}(t-)\d \overline{\Phi}(t), \quad \overline{\kappa}(0)=1,
\end{align}
is the value process of a self-financing investment portfolio with respect to $\overline{\Phi}$. 

Furthermore, given the valuation basis $(\overline{\Phi}, \overline{\Lambda})$, let $\overline{p}=(\overline{p}(s,t))_{0\leq s\leq t}$ with $\overline{p}(s,t)=(\overline{p}_{jk}(s,t))_{jk}$ denote the solution of the stochastic differential equation system
\begin{align}\label{kolmforward}
\overline{p}_{jk}(s,\d t)=\sum\limits_{i} \overline{p}_{ji}(s,t-)\d \overline{\Lambda}_{ik}(t), \ \overline{p}_{jk}(s,s)=\delta_{jk}, \ t>s.
\end{align}
Observe that  we may pick $N$ itself for $\overline{\Lambda}$. In this case the solution of \eqref{kolmforward} satisfies $p_{aj}(0,t)=I_j(t)$, since $I_j(0)=
\delta_{aj}$ and
\begin{align}
\d I_j(t)=\sum\limits_{k:k\neq j} (\d N_{kj}(t)-\d N_{jk}(t))=\sum\limits_{k} I_k(t-)\d N_{kj}(t).
\end{align}
\indent Throughout this paper, let the valuation basis $(\Phi,\Lambda)$ represent the so-called \emph{second order valuation basis}. The process $\Phi$ describes the real return in ivestment in the insurer's investment portfolio.  Let $\kappa$ denote the solution of \eqref{SDEPhi} with respect to $\Phi$. 
For the second-order basis we additionally assume that
\begin{itemize}
\item $\Lambda$ is a predictable process,
\item conditional on $(\Phi,\Lambda)=(E,F)$ the process $Z$ is a Markov process under $\mathbb{P}$ with cumulative transitions intensity matrix  $F$.
\end{itemize}
So the process $I_j(t-) \d \Lambda_{jk}(t)$ is a $\mathbb{P}$-compensator of $\d N_{jk}$ with respect to the natural completed filtration of the random vector $(Z^t,\Phi,\Lambda)_{t\geq 0}$. Due to the conditional Markov property, the stochastic differential equation (\ref{kolmforward})  with respect to $\Lambda$ corresponds to the Kolmogorov forward equation of $Z$  conditional on $(\Phi,\Lambda)$, and its solution  $p(s,t)=(p_{jk}(s,t))_{jk}$ is the transition probability matrix of $Z$ conditional on $(\Phi,\Lambda)$.

Furthermore, let the valuation basis $(\Phi^*,\Lambda^*)$  represent  the so-called \emph{first order valuation basis}.
For this specific  valuation basis we additionally assume  that
\begin{itemize}
\item $\Phi^*$ and $\Lambda^*$ are deterministic,
\item $Z$ is a Markov process under a prudent probability measure $\mathbb{P}^*$ with cumulative transition intensities $\Lambda^*_{jk}$, $j \neq k$,
\item $(\mathbb{I}+\Delta \Lambda^*_M(t))^{-1}$ exists for every $t>0$,
\end{itemize}
where $\Lambda^*_M$ denotes the matrix-valued process  $\Lambda^*_M = (\Lambda^*_{jk})_{jk}$ with $\Lambda^*_{jj}\coloneqq -\sum_{k:k\neq j} \Lambda^*_{jk}$.
Let  $\kappa^*$ and $p^*$ be the solutions of \eqref{SDEPhi} and \eqref{kolmforward} with respect to $\Phi^*$ and $\Lambda^*$, respectively.  Under the first order valuation,   \eqref{kolmforward} is the classical Kolmogorov forward equation and $p^*$ is the classical transition probability matrix of $Z$ under $\mathbb{P}^*$. The existence of $(\mathbb{I}+\Delta \Lambda^*_M(t))^{-1}$ for every $t>0$ ensures that the matrix $p^*(s,t)$ has an inverse for each $s\leq t$, denoted as $q^*(s,t)$, cf.~Lemma \ref{itoinverse} in the appendix. In particular, $q^*$ satisfies the stochastic differential equation
\begin{align*}
q^*(s,\d t)=-(\d G(t))q^*(s,t-),\ q^*(s,s)=\mathbb{I},\ t>s,
\end{align*}
where $G(t)=\Lambda^*_M(t)-\sum_{0<s\leq t} (\Delta \Lambda^*_M(s))^2(\mathbb{I}+\Delta \Lambda^*_M(s))^{-1}$ (cf. Lemma \ref{itoinverse}).

Recall that the insurance policy shall have a finite contract horizon in $[0,T]$.
We assume that the insurance cash flow $B$ has the form
\begin{align}\label{DefOfBindiv}
  \d B (t) =  \sum_j I_j(t-)\, \d B_j(t) + \sum_{jk:j \neq k} b_{jk}(t)\, \d N_{jk}(t),
\end{align}
where $(B_j)_j$ are right-continuous finite variation functions  that satisfy $\d B_j(t)=0$ for $t>T$,  and $(b_{jk})_{jk:j \neq k}$ are bounded and measurable functions with  $b_{jk}(t)=0 $ for $t>T$.
\newpage 
We generally assume that
\begin{itemize}
\item the processes $\Phi^*$, $\Phi$ and $(N, \Lambda^*,\Lambda,(B_j)_j)$ have no simultaneous jumps.
\end{itemize}
The latter condition implies that the covariation between the investment risk and all other risk drivers is zero. This fact will help us to build additive decompositions by applying It\^{o}'s formula, cf.~Lemma \ref{lemmasurplusdecomp} below.

\subsection*{Individual revaluation surplus}

In with-profit life insurance, the remaining future liabilities of the individual insurance contract  at time $t$ are commonly evaluated as
\begin{align*}
 \sum\limits_{j} I_j(t)V_j^*(t),
\end{align*}
where $V_j^*(t)$ shall be the prospective reserve at time $t$ in state $j$ with respect to the first order valuation basis, cf.~Norberg (1999). According to  Milbrodt \& Helbig (1999, chapter 10.A) it holds that
\begin{align*}
V_j^*(t)=\sum\limits_{k} \int_{(t,T]} \frac{\kappa^*(t)}{\kappa^*(s)}p^*_{jk}(t,s-)\d B_k(s)+\sum\limits_{k,l: k\neq l} \int_{(t,T]}  \frac{\kappa^*(t)}{\kappa^*(s)} p^*_{jk}(t,s-)b_{kl}(s)d\Lambda^*_{kl}(s).
\end{align*}
The accrued assets of the individual insurance contract at time $t$ equal \eqref{DefA}, so the total surplus of the individual policy at time $t$ is
\begin{align}\label{DefSind}
S(t)=-\int_{[0,t]} \frac{\kappa(t)}{\kappa(s)}\d B(s)-\sum\limits_{j} I_j(t)V_j^*(t),
\end{align}
cf.~Norberg (1999). 
The corresponding revaluation process $R$ equals
\begin{align}\label{DefRind}
R(t)= \frac{S(t)}{\kappa(t)}= -\int_{[0,t]} \frac{1}{\kappa(s)}\d B(s)-\sum\limits_{j} \frac{1}{\kappa(t)} I_j(t)V_j^*(t).
\end{align}

\begin{proposition}\label{propositionfunctionalindsurplus} For $R$ defined by \eqref{DefRind} and  $t \in [0,T]$  it holds that
\begin{align}\label{rhoindsurplus}
R(t)=-H((\Phi^*, \Lambda^*)+  (\Phi-\Phi^*,N-\Lambda^*)^t),
\end{align}
where for any valuation basis $(\overline{\Phi}, \overline{\Lambda})$ the mapping $H$ is defined by
\begin{align}\label{functionaldef}
H((\overline{\Phi}, \overline{\Lambda}))\coloneqq \sum\limits_{j} \int_{[0,T]} \frac{1}{\overline{\kappa}(s)}\overline{p}_{aj}(0,s-)\d B_j(s)+\sum\limits_{j,k:j\neq k} \int_{(0,T]} \frac{1}{\overline{\kappa}(s)} \overline{p}_{aj}(0,s-)b_{jk}(s)d\overline{\Lambda}_{jk}(s)
\end{align}
with $\overline{p}_{aj}(0,0-)\coloneqq \delta_{aj}$. %
\end{proposition}
\begin{proof} The solution of (\ref{kolmforward}) with respect to the cumulative transition intensity vector $\Lambda^* + (N-\Lambda^*)^t$ is
\begin{align*}
\begin{cases} I_j(s), \ s\leq t, \\ \sum\limits_{k} I_k(t)p^*_{kj}(t,s), \ t>s, \end{cases}
\end{align*}
where $p^*_{kj}(t,s)$ is the solution of (\ref{kolmforward}) with respect to the first order valuation basis. The solution of (\ref{SDEPhi}) with respect to $\Phi^*+(\Phi-\Phi^*)^t$ is
\begin{align*}
\begin{cases}
\kappa(s), \ s\leq t, \\ \kappa(t)\frac{\kappa^*(s)}{\kappa^*(t)},\ s>t,
\end{cases}
\end{align*}
where $\kappa^*$ is the solution of (\ref{SDEPhi}) with respect to the first order valuation basis. By plugging these solutions into (\ref{functionaldef}), we obtain the desired result.
\end{proof}
Proposition \ref{propositionfunctionalindsurplus} allows us to represent $R$ by
\begin{align*}
  R(t) = \varrho (X^t) , \quad t \geq 0,
\end{align*}
 for  various  choices of $X$ and  $\varrho$.  For example, we may  define the risk basis $X$ and the mapping $\varrho$ as  follows:
\begin{example}\label{ExpRX1}
  By setting
  $$X=(X_{\Phi},X_u,X_s)=(\Phi-\Phi^*,N-\Lambda,\Lambda-\Lambda^*),$$
   we distinguish between financial risk, unsystematic biometric risk and systematic biometric risk,   and we may define $\varrho$ by     $$\varrho(X^t) = -H\big((\Phi^*,\Lambda^*)+ (X_{\Phi}^t,X_u^t+X_s^t) \big).$$
\end{example}
\begin{example}\label{ExpRX2}
By setting $$X=(X_{\Phi},(X_{jk})_{jk: j\neq k} )= (\Phi-\Phi^*,(N_{jk}-\Lambda^*_{jk})_{jk: j\neq k}),$$ we distinguish between financial risk and transition-wise biometric risks,  and we may define $\varrho$ by
   $$     \varrho(X^t) = -H\big((\Phi^*,\Lambda^*)+ (X_{\Phi}^t, (X_{jk}^t)_{jk:j\neq k}) \big).$$
\end{example}
\begin{example} \label{ExpRX3} Let the processes $(\Phi_j)_j$ and $(\Phi^*_j)_j$ be defined by $ \d \Phi_j (t) = I_j(t-) \d\Phi(t)$, $\Phi_j(0)=0$, and $\d \Phi^*_j (t) = I_j(t-) \d\Phi^*(t)$, $\Phi^*_j(0)=0$, respectively. Further, we denote $\Lambda_j=(\Lambda_{jk})_{k:k \neq j}$ and $\Lambda^*_j=(\Lambda^*_{jk})_{k:k \neq j}$.
By setting $$X=(X_u,(X_j)_j) =(X_u, (X_{j,1},X_{j,2})_j) =(N-\Lambda,  (\Phi_j-\Phi^*_j,  \Lambda_{j}-\Lambda^*_{j})_j)$$ we distinguish  between unsystematic biometric risk and state-wise remaining risks,  and we may define  $\varrho$ by $$     \varrho(X^t) = -H\left((\Phi^*,\Lambda^*_{jk})+ (0,X_u^t)+ \left(\sum_j X_{j,1}^t,(X_{j,2}^t)_j\right)\right).$$
\end{example}

\subsection*{Mean portfolio revaluation surplus}
In actuarial practice it is not uncommon to focus on mean portfolio values only. We can replicate this perspective by applying the expectation $\E[ \,\cdot\, | \Phi, \Lambda]$ on the individual values.
 The resulting mean portfolio revaluation surplus is
\begin{align}\label{DefRcoll}
 R'(t) = \E[ R(t) | \Phi, \Lambda] ,
\end{align}
and the corresponding mean portfolio total surplus  is
\begin{align}\label{DefScoll}
 S'(t) = \E[ S(t) | \Phi, \Lambda]= \kappa(t) R'(t).
\end{align}
Note that Norberg (1999) uses the definition $S'(t) = \E[ S(t) | \Phi^t, \Lambda^t] $ instead, but this definition is equivalent since
\begin{align*}
&\E[ S(t) | \Phi, \Lambda] \\
&= - \int_{[0,t]} \frac{\kappa(t)}{\kappa(s)} \sum_j \bigg( p_{aj}(0,s-) \d B_j(s) + \sum_{k:k \neq j} b_{jk}(s) p_{aj}(0,s-) \d\Lambda_{jk}(s)  \bigg)-\sum\limits_{j} p_{aj}(0,t)V_j^*(t)
\end{align*}
is $\sigma( \Phi^t,\Lambda^t)$-measurable.
The following corollary is a direct consequence of Proposition \ref{propositionfunctionalindsurplus}.
\begin{corollary}\label{propositionfunctionalcollsurplus}
 For $R'$ defined by \eqref{DefRcoll} and  $t \in [0,T]$  it holds that
\begin{align}\label{rhomeansurplus}
R'(t)= \E\big[ - H((\Phi^*,\Lambda^*) + (\Phi-\Phi^*, N-\Lambda^*)^t)\big| \Phi, \Lambda\big].
\end{align}
\end{corollary}
Because of the latter corollary, in the examples \ref{ExpRX1} to \ref{ExpRX3} we just need to add the
conditional expectation $\E[ \,\cdot\, | \Phi, \Lambda]$ to the definition of $\varrho$ in order to get to the  mean portfolio perspective.  The next example is in particular relevant in German life insurance.
\begin{example}\label{Example:RinGermany}
Consider a life insurance contract with the states active, surrendered and dead, $$\mathcal{Z}=\{a,s,d\}.$$
We assume that $\Lambda^*$ and $\Lambda$ are absolutely continuous with densities $\lambda^*$ and $ \lambda$, respectively. Let
\begin{align*}
  \phantom{.}_{k-l}p^*_{x+l}=p^*_{aa}(l,k), \qquad   q^*_{x+k-1}=p^*_{ad}(k-1,k), \qquad r^*_{x+k-1}=p^*_{as}(k-1,k).
\end{align*}
We assume that sojourn payments occur only in state active and only as lump sum payments $b_k$ at integer times $k$. Furthermore, we assume  that the death benefit function  and the surrender benefit function  have the form
 \begin{align*}
 b_{ad}(t) = \frac{\kappa(\lfloor t \rfloor)}{\kappa(t)} d_{\lceil t \rceil},\qquad b_{as}(t) = \frac{\kappa(\lfloor t \rfloor)}{\kappa(t)} s_{\lceil t \rceil},
 \end{align*}
 where  $d_{\lceil t \rceil}$ and $s_{\lceil t \rceil}$ represent  the death benefit and surrender benefit in year $\lfloor t \rfloor$. This definition of $b_{ad}$ and $b_{as}$ discounts death benefits and surrender benefits as if they are paid out at the end of the year, so that $V^*_a$ has at  integer times $l$ the representation
  \begin{align*}
     V_a^*(l)= \sum_{k=l+1}^{T} \frac{\kappa^*(l)}{\kappa^*(k)} \,\phantom{.}_{k-l}p^*_{x+l} \,b_k  + \sum_{k=l+1}^T \frac{\kappa^*(l)}{\kappa^*(k)} \phantom{.}_{k-l-1}p^*_{x+l}  \big(d_k\,  q^*_{x+k-1}   + s_k\, r^*_{x+k-1}    \big).
  \end{align*}
We  define yearly interest rates of first order and second order by
\begin{align*}
   i^*_k =  e^{\int_k^{k+1} \phi^*(u)\, \d u }-1 , \qquad  i_k =  e^{\int_k^{k+1} \phi(u)\, \d u }-1 , \qquad k \in \mathbb{N}_0.
\end{align*}
One can show that the  yearly increments of the mean portfolio revaluation surplus process equal
\begin{align*}
 &R'(k+1) - R'(k)\\
 &=  \mathrm{e}^{-\int_0^{k+1} \phi(u) \, \d u}\phantom{.}_{k}p_{x}   \Big( V_a^*(k)\, (1+i_k)  - q_{x+k}\, d_{k+1}    - r_{x+k} \, s_{k+1} - p_{x+k}\,\big( b_{k+1}+V^*_a(k+1)\big)\Big) .
\end{align*}
This formula is commonly used in German life insurance, cf.~Milbrodt \& Helbig (1999, section 11.B). It is  common in Germany to decompose the increments $R'(k+1) - R'(k)$ into investment surplus, mortality surplus and lapse surplus. For that purpose, analogously to Example \ref{ExpRX2} we choose $$X=(X_{\Phi},X_{ad},X_{as})=(\Phi-\Phi^*,N_{ad}-\Lambda^*_{ad},N_{as}-\Lambda^*_{as})$$
as risk basis.
\end{example}

\section{The ISU decomposition principle}

Recall that the $t$-stopped process $X^{t}=(X^t_1, \ldots, X_m^t)$  represents the currently available information on the risk factors $X_1, \ldots, X_m$  at time $t$.
Suppose that the  information updates of the risk factors $X_1, \ldots, X_m$ are asynchronously delayed  with $t_1, \ldots , t_m  \leq t$ being the current update statuses of each risk factor.
Then
\begin{align*}
U(t_1, \ldots, t_m):=    \varrho((X^{t_1}_1, \ldots, X^{t_m}_m) )
\end{align*}
is the value of the delayed revaluation process at time $t$. %
 We denote $U=( U(t_1, \ldots, t_m))_{ t_1, \ldots , t_m  \geq 0}$
as the \emph{revaluation surplus surface} with respect to $X$.
We can recover the revaluation surplus process  $R$ from the revaluation surplus surface $U$ as
\begin{align*}
  R(t) = U(t, \ldots, t), \quad  t\geq 0.
\end{align*}
For any  partition $\mathcal{T}(t) = \{0=t_0 < t_1 < \cdots < t_k =t \}$ of the interval $[0,t]$ we
  can build the telescoping series
\begin{align*}
 R(t)-R(0) &= U(t, \ldots, t) -U(0, \ldots,0)\\
  &= \sum_{l=0}^{k-1} \Big( U(t_{l+1},t_{l}, \ldots,t_{l})-U(t_{l},t_{l}, \ldots,t_{l})\Big) \\
   &\quad +  \sum_{l=0}^{k-1} \Big( U(t_{l+1},t_{l+1}, t_{l},\ldots,t_{l})-U(t_{l+1},t_{l}, \ldots,t_{l})\Big) \\
  &\quad + \cdots \\
  & \quad + \sum_{l=0}^{k-1}  \Big( U(t_{l+1},\ldots, t_{l+1}, t_{l+1})-U(t_{l+1},\ldots, t_{l+1}, t_{l})\Big) .
 \end{align*}
 It is natural here to interpret the $m$ different sums on the right hand side as an additive  decomposition $R(t)-R(0)=D_1(t)+  \cdots + D_m(t)$,   since the $i$-th sum collects exactly the  information updates for the $i$-th risk factor.

\begin{definition}
The  random vector   $D(t)=(D_1(t), \ldots, D_m(t))$  defined by
\begin{align}\label{SUdecomposition}\begin{split}
  D_1(t) &= \sum_{l=0}^{k-1} \Big( U(t_{l+1},t_{l}, \ldots,t_{l})-U(t_{l},t_{l}, \ldots,t_{l})\Big), \\
  & \cdots  \\
    D_{m}(t) &= \sum_{l=0}^{k-1 }  \Big( U(t_{l+1},\ldots, t_{l+1}, t_{l+1})-U(t_{l+1},\ldots, t_{l+1}, t_{l})\Big),
 \end{split}\end{align}
  is called  the \emph{SU (sequential updating) decomposition}  of $R(t)-R(0)$ with respect to $\mathcal{T}(t)$.
\end{definition}
The SU decomposition principle is used in various fields of economics, see for example Fortin et al.~(2011) and Biewen (2014). In the definition formula \eqref{SUdecomposition} we update  the information on  $X$  in a specific order, starting with risk factor $X_1$, then updating $X_2$, and so on.  Unfortunately, the decomposition is  not invariant with respect to this update order, which is  a major drawback of the SU concept.  We can reduce the impact of the update order by increasing the number of updating steps, i.e.~refining the partition $\mathcal{T}_n(t)$. In a next step we push such refinements to the limit.

 Let  $\mathcal{T}_n(t)=\{ 0=t^n_0 < t^n_1 < \cdots  < t^n_{k_n} =t \}$, $n \in \mathbb{N}$, be a sequence of
 partitions of $[0,t]$  with vanishing step lengths (i.e.~$  \lim_{ n\rightarrow \infty} \max_{1 \leq l  \leq k_{n}} |t^n_{l}-t^n_{l-1}| =0$).
 For each $n \in \mathbb{N}$ let  $D^n(t)=(D_1^{n}(t), \ldots ,D_{m}^{n}(t))$ be the SU decomposition of $R(t)-R(0)$ with respect to $\mathcal{T}_n(t)$. We are looking for a random vector  $D(t)$ that satisfies
\begin{align}\label{ISUdecomposition}\begin{split}
 D_i(t) = \plim_{n \rightarrow \infty} D^{n}_i(t), \quad i \in \{0, \ldots, m\}.
 \end{split}\end{align}
\begin{definition}
 Let $(\mathcal{T}_n(t))_{n \in \mathbb{N}}$ be a sequence of partitions of $[0,t]$ with vanishing step lengths.
 If  $D(t)$ satisfies \eqref{ISUdecomposition}, then we call $D(t)$ the \emph{ISU (infinitesimal sequential updating) decomposition}  of $R(t)-R(0)$ with respect to $(\mathcal{T}_n(t))_{n \in \mathbb{N}}$.
\end{definition}

\section{ISU decompositions in multi-state life insurance}

This section contains general technical results that will be needed for the examples in the next section. The proofs can be found in the appendix. For any valuation basis $(\overline{\Phi},\overline{\Lambda})$, we write 
\begin{align*}
\widetilde{\overline{\Phi}}(t)&=\overline{\Phi}(t)-[\overline{\Phi},\overline{\Phi}]^c(t)-\sum_{0<s\leq t} (1+\Delta \overline{\Phi}(s))^{-1} (\Delta \overline{\Phi}(s))^2.
\end{align*}

Moreover, let $R^*_{jk}$, $j\neq k$ denote the first order sum at risk, i.e.
\begin{align*}
R^*_{jk}(t)=b_{jk}(t)+V_k^*(t)-V_j^*(t).
\end{align*}

\begin{lemma}\label{lemmasurplusdecomp} Let $(\overline{\Phi}, \overline{\Lambda})$ be a valuation basis such that $(\Phi^*,\overline{\Phi})$ and $(\Lambda^*,\overline{\Lambda},(B_j)_j)$ have no simultaneous  jumps. Then it holds that
\begin{align*}
H\big((\Phi^*,\Lambda^*)+(\overline{\Phi}-\Phi^*, \overline{\Lambda}-\Lambda^*)^t\big)
&=\int_{(0,t]}\frac{1}{\overline{\kappa}(s-)}\sum\limits_{j} \overline{p}_{aj}(0,s-) V^*_j(s-)\d (\widetilde{\overline{\Phi}}-\Phi^*+[\widetilde{\overline{\Phi}},\Phi^*])(s) \\
& \quad -  \sum_{jk: j \neq k}\int_{(0,t]} \frac{1}{\overline{\kappa}(s)} \overline{p}_{aj}(0,s-)R^*_{jk}(s)\d (\overline{\Lambda}_{jk}-\Lambda^*_{jk})(s).
\end{align*}
\end{lemma}
\begin{theorem}\label{Theorem:ISU1}  Let the processes $(\Phi_j)_j$ and $(\Phi^*_j)_j$ be defined by $ \d \Phi_j (t) = I_j(t-) \d\Phi(t)$, $\Phi_j(0)=0$, and $\d \Phi^*_j (t) = I_j(t-) \d\Phi^*(t)$, $\Phi^*_j(0)=0$, respectively.  For $j,k \in \mathcal{Z}$ let
\begin{align*}
   X_{\Phi,j}(t) &=\Phi_j-\Phi^*_j,\\
   X_{u,jk}(t) &= N_{jk}-\Lambda_{jk},\\
  X_{s,jk}(t) &= \Lambda_{jk}-\Lambda^*_{jk},
\end{align*}
and set $X=((X_{\Phi,j})_j,(X_{u,jk})_{j,k:j \neq k},(X_{s,jk})_{j,k:j\neq k})$. Then
\begin{align*}
  \varrho (X^t)= -H\left((\Phi^*,\Lambda^*)+\left( \sum_j X_{\Phi,j} , ( X_{u,jk} +X_{s,jk} )_{jk: j\neq k}\right)^t\right)
\end{align*}
has the ISU decomposition
  \begin{align*}
D_{\Phi,j}(t)&=\int_{(0,t]}\frac{1}{\kappa(s-)} I_j(s-) V_j^*(s-)\d (\widetilde{\Phi}-\Phi^*)(s), \\
D_{u,jk}(t)&= -  \int_{(0,t]} \frac{1}{\kappa(s)} I_j(s-)R_{jk}^*(s)\d (N_{jk} -\Lambda_{jk} )(s),\\
D_{s,jk}(t) &= -  \int_{(0,t]} \frac{1}{\kappa(s)} I_j(s-)R_{jk}^*(s)\d (\Lambda_{jk}-\Lambda_{jk}^*)(s).
\end{align*}
In particular, the ISU decomposition does not depend on the update order.
\end{theorem}

\begin{lemma}\label{ISU:ExpectationConsistency}
Let $X=(X_1, \ldots , X_m)$ be a given risk basis with
$$R(t) = \varrho ((X_1, \ldots ,X_m)^t)$$ for a suitable mapping $\varrho$, generating the ISU decomposition $D(t)=(D_1(t), \ldots, D_m(t))$ with respect to $(\mathcal{T}_n(t))_n$, and let $\mathcal{G}$ be a sub-$\sigma$-algebra of $ \mathcal{A}$.
Suppose that the SU decomposition $D^n(t)=(D^n_1(t),\ldots,D_m^n(t))$ of $R(t)-R(0)$ with respect to $\mathcal{T}_n(t)$ satisfies $|D^n_i(t)|\leq Y$, $i=1,\ldots,m$, $n\in\mathbb{N}$, for some integrable random variable $Y$.
Then the ISU decomposition of
 $$\widetilde{R}(t)=\widetilde{\varrho}((X_1, \ldots ,X_m)^t)\coloneqq \mathbb{E}\left[\varrho((X_1, \ldots ,X_m)^t)|\mathcal{G}\right]$$
 is given by $$\widetilde{D}(t)=(\mathbb{E}[D_1(t)|\mathcal{G}], \ldots, \mathbb{E}[D_m(t)|\mathcal{G}]).$$
\end{lemma}
\begin{proof}
  Since  the revaluation surplus surfaces $U$ and  $\widetilde{U}$ are linked via the equation
\begin{align*}
	\widetilde{U}(t_1, \ldots ,t_m) = \mathbb{E}[U(t_1,\ldots ,t_m)|\mathcal{G}],
\end{align*}
the SU decomposition of $\widetilde{R}(t)-\widetilde{R}(0)$ is given by $\widetilde{D}^n(t)=(E[D_1^n(t)|\mathcal{G}], \ldots, E[D_m^n(t)|\mathcal{G}])$.
Using that $|D_i^n(t)|\leq Y$, $i=1,\ldots,m$, for some integrable random variable $Y$, the Dominated Convergence Theorem for conditional expectations almost surely yields
$$\widetilde{D}_i(t)=\lim\limits_{n\to \infty} \mathbb{E}[D^n_i(t)|\mathcal{G}]=E[D_i(t)|\mathcal{G}], \; i=1,\ldots,m.$$
\end{proof}

\begin{theorem}\label{Theorem:ISU2} Let $X$ be defined as in Theorem \ref{Theorem:ISU1}.  Then
\begin{align*}
  \varrho (X^t)= \E\left[ -H\left((\Phi^*,\Lambda^*)+\left( \sum_j X_{\Phi,j} , ( X_{u,jk} +X_{s,jk} )_{jk: j\neq k}\right)^t\right) \Bigg| \Phi, \Lambda \right]
\end{align*}
has the ISU decomposition
  \begin{align*}
D_{\Phi,j}(t)&=\int_{(0,t]}\frac{1}{\kappa(s-)} p_{aj}(0,s-) V_j^*(s-)\d (\widetilde{\Phi}-\Phi^*)(s), \\
D_{u,jk}(t)&= 0,\\
D_{s,jk}(t) &= -  \int_{(0,t]} \frac{1}{\kappa(s)} p_{aj}(0,s-) R_{jk}^*(s)\d (\Lambda_{jk}-\Lambda_{jk}^*)(s).
\end{align*}
In particular, the ISU decomposition does not depend on the update order.
\end{theorem}

\begin{proposition} \label{AggregationOfISUs}
Let $X=(X_1, \ldots , X_m)$ be a given risk basis with
$$R(t) = \varrho ((X_1+X_2, X_3, \ldots ,X_m)^t)$$ for a suitable mapping $\varrho$, generating the ISU decomposition $D(t)=(D_1(t), \ldots, D_m(t))$. Then the partially aggregated risk basis
$$\widetilde{X}= (X_1+X_2, (X_3,X_4),X_5 \ldots ,X_m)$$  generates the ISU decomposition
$$\widetilde{D}(t)=(D_1(t)+D_2(t), D_3(t)+D_4(t),D_5(t) \ldots, D_m(t)).$$
\end{proposition}
\begin{proof}
  Since  the revaluation surplus surfaces $U$ and  $\widetilde{U}$ are linked via the equation
 \begin{align*}
   \widetilde{U}(t_1, t_3, t_5 \ldots ,t_m) = U(t_1,t_1, t_3, t_3, t_5\ldots ,t_m),
 \end{align*}
 the SU decompositions $D^n$ and $\widetilde{D}^n$ with respect to $\mathcal{T}_n(t)$ satisfy $$\widetilde{D}^n(t)=(D^n_1(t)+D^n_2(t), D^n_3(t)+D_4^n(t), D_5^n(t), \ldots, D^n_m(t)).$$ The latter equation carries through the limit \eqref{ISUdecomposition} to the ISU decompositions.
\end{proof}

\section{Examples}\label{sectionExamples}

We continue with the examples for the risk basis $X$ and the mapping $\varrho$ from section \ref{Section:LifeInsuranceModel} and present the corresponding ISU decompositions.

\subsection*{Decomposition of the individual revaluation surplus}

Let $R$ be the individual revaluation surplus according to \eqref{DefRind}.
\begin{example}\label{IndSurplDec1}
	Suppose that we are in the setting of Example \ref{ExpRX1}, where
	we distinguish between  financial risk, unsystematic biometric risk  and systematic biometric risk. 
	By applying Theorem \ref{Theorem:ISU1} and Proposition \ref{AggregationOfISUs} we obtain the ISU decomposition
	\begin{align*}
		D_{\Phi}(t)&=\int_{(0,t]}\frac{1}{\kappa(s-)} \sum_j I_j(s-) V_j^*(s-)\d (\widetilde{\Phi}-\Phi^*)(s), \\
		D_{u}(t)&= -  \sum_{jk:j \neq k} \int_{(0,t]} \frac{1}{\kappa(s)}I_j(s-) R_{jk}^*(s)\d (N_{jk}-\Lambda_{jk})(s),\\
		D_{s}(t) &= - \sum_{jk:j \neq k} \int_{(0,t]} \frac{1}{\kappa(s)} I_j(s-)R_{jk}^*(s)\d (\Lambda_{jk}-\Lambda_{jk}^*)(s).
	\end{align*}
\end{example}
\begin{example}
	Suppose that we are in the setting of Example \ref{ExpRX2}, where we distinguish between financial risk and  transition-wise biometric risks. 
	By applying Theorem \ref{Theorem:ISU1} and Proposition \ref{AggregationOfISUs} we obtain the ISU decomposition
	\begin{align*}
		D_{\Phi}(t)&=\int_{(0,t]}\frac{1}{\kappa(s-)} \sum_j I_j(s-) V_j^*(s-)\d (\widetilde{\Phi}-\Phi^*)(s), \\
		D_{jk}(t)&= -   \int_{(0,t]} \frac{1}{\kappa(s)}I_j(s-) R_{jk}^*(s) \d (N_{jk}-\Lambda^*_{jk})(s),\quad j,k \in \mathcal{Z}, \, j \neq k.
	\end{align*}
	As a special case this ISU decomposition includes the heuristic approach of Ramlau-Hansen (1988, formula (4.7)) for subdividing  biometric surplus in a transition-wise way.
\end{example}

\begin{example}\label{IndSurplDec3}
	Suppose that we are in the setting of Example \ref{ExpRX3}, where
	we distinguish unsystematic biometric risk and state-wise remaining risks. By applying Theorem \ref{Theorem:ISU1} and Proposition \ref{AggregationOfISUs} we obtain the ISU decomposition
	\begin{align*}
		D_{u}(t)&= -  \sum_{jk:j \neq k} \int_{(0,t]} \frac{1}{\kappa(s)}I_j(s-) R_{jk}^*(s) \d (N_{jk}-\Lambda_{jk})(s),\\
		D_{j}(t)&=\int_{(0,t]}\frac{1}{\kappa(s-)} I_j(s-)\Big(  V_j^*(s-)\d (\widetilde{\Phi}-\Phi^*)(s)
		- \sum_{k:k \neq j} R_{jk}^*(s)\d (\Lambda_{jk}-\Lambda_{jk}^*)(s)\Big),  \quad j \in \mathcal{Z}.
	\end{align*}
	As a special case this ISU decomposition includes heuristic approaches of Ramlau-Hansen (1988, formula before (4.10)) and Norberg (1999, formula (5.4)) for splitting off unsystematic biometric surplus and then subdividing the remaining surplus in a state-wise way.
\end{example}

In Example \ref{IndSurplDec1} and Example \ref{IndSurplDec3} we split off the surplus contribution of the unsystematic biometric risk.  Since this unsystematic biometric risk is diversifiable in the insurance portfolio, its contribution  $\kappa(t) D_u(t)$ to the total surplus $S(t)$, cf.~\eqref{totalSurplus2}, is typically credited or debited to the insurer.
M{\o}ller \& Steffensen (2007, chapter 6.3) denote the remaining surplus $S(t) - \kappa(t) D_u(t)$ as the 'systematic surplus'. This systematic surplus mainly belongs to the policyholder.

Asmussen \& Steffensen (2020, chapter VI.4) split also the financial risk into an unsystematic part and a systematic part and argue that the  unsystematic financial risk surplus contribution should be fully credited or debited to the insurer, similarly to the unsystematic biometric risk surplus contribution. They distinguish unsystematic and systematic financial risk by splitting
$\Phi$ into a martingale part and a remaining systematic part. If we  likewise split $\Phi-\Phi^*$ in the risk basis $X$ into a martingale part and a remaining systematic part, then the resulting ISU decomposition allows us to distinguish between systematic and unsystematic surplus contributions. If we then collect the systematic biometrical and systematic financial surplus contributions, then we just end up with the systematic surplus formula of Asmussen \& Steffensen (2020, chapter VI.4). We do not show the detailed calculations here  but leave them  to the reader.

\subsection*{Decomposition of the mean portfolio revaluation surplus}

Let $R'$ be the mean portfolio revaluation surplus according to \eqref{DefRcoll}.

\begin{example}
	We choose the setting from Example \ref{ExpRX1} but  adopt the mean portfolio perspective.
	By applying Theorem \ref{Theorem:ISU2} and Proposition \ref{AggregationOfISUs} we obtain the ISU decomposition
	\begin{align*}
		D_{\Phi}(t)&=\int_{(0,t]}\frac{1}{\kappa(s-)} \sum_j p_{aj}(0,t-) V_j^*(s-)\d (\widetilde{\Phi}-\Phi^*)(s), \\
		D_{u}(t) &= 0,\\
		D_{s}(t) &= - \sum_{jk:j \neq k} \int_{(0,t]} \frac{1}{\kappa(s)} p_{aj}(0,t-) R_{jk}^*(s)\d (\Lambda_{jk}-\Lambda_{jk}^*)(s).
	\end{align*}
The conditional expectation in \eqref{DefRcoll} and  \eqref{DefScoll} completely eliminates the unsystematic biometric risk, which explains why  we have $D_{u}(t) = 0$ here.
\end{example}
\begin{example}
	Here we choose the setting from Example \ref{ExpRX2} but adopt the mean portfolio perspective.
	By applying Theorem \ref{Theorem:ISU2} and Proposition \ref{AggregationOfISUs} we obtain the ISU decomposition
	\begin{align*}
		D_{\Phi}(t)&=\int_{(0,t]}\frac{1}{\kappa(s-)} \sum_j p_{aj}(0,t-) V_j^*(s-)\d (\widetilde{\Phi}-\Phi^*)(s), \\
		D_{jk}(t)&= -   \int_{(0,t]} \frac{1}{\kappa(s)}p_{aj}(0,t-) R_{jk}^*(s) \d (\Lambda_{jk}-\Lambda_{jk}^*) (s), \quad j,k \in \mathcal{Z}, \, j \neq k.
	\end{align*}
The next example shows an application of this formula.
\end{example}

\begin{example}
	We continue with the previous example  but focus here on the specific setting of Example
	\ref{Example:RinGermany}. 
	One can show  that the SU decomposition of $R'(k+1)-R'(k)$ with respect to an integer  partition equals
	\begin{align}\begin{split}\label{SUGermanLI}
			U(k+1,k,k)-U(k,k,k) &=  e^{-\int_0^{k+1} \phi(u)\, \d u } \phantom{.}_kp_{x} \,V^*_a(k)\,\big(i_k-i_k^*\big), \\
			U(k+1,k+1,k)-U(k+1,k,k) & = e^{-\int_0^{k+1} \phi(u)\, \d u } \phantom{.}_kp_{x}  \big(V^*_a(k+1-)-d_{k+1}\big)  \,\big(q_{x+k} -q^*_{x+k}\big),\\
			U(k+1,k+1,k+1)-U(k+1,k+1,k) & =  e^{-\int_0^{k+1} \phi(u)\, \d u } \phantom{.}_kp_{x}  \big(V^*_a(k+1-)-s_{k+1} \big)  \,\big(r_{x+k} -r^*_{x+k}\big),
	\end{split}\end{align}
	see the appendix.
	This decomposition is the standard surplus decomposition formula used German life insurance, cf.~Milbrodt \& Helbig (1999, section 11.B).  We can interpret the latter SU decomposition as an approximation of the ISU decomposition of $R'(k+1)-R'(k)$, which equals here
	\begin{align}\begin{split}\label{ISUGermanLI}
			D_{\Phi}(k+1)-D_{\Phi}(k) &= \int_{(k,k+1]} e^{-\int_0^{s} \phi(u)\, \d u } \phantom{.}_sp_{x}\,
			V^*_a(s) \, \d (\Phi-\Phi^*)  (s), \\
			D_{ad}(k+1)-D_{ad}(k) & = \int_{(k,k+1]} e^{-\int_0^{s} \phi(u)\, \d u } \phantom{.}_sp_{x}    \big( V^*_a(s)-b_{ad}(s)  \big) \, \d (\Lambda_{ad} -\Lambda^*_{ad})(s),\\
			D_{as}(k+1)-D_{as}(k) & = \int_{(k,k+1]}  e^{-\int_0^{s} \phi(u)\, \d u } \phantom{.}_sp_{x}   \big( V^*_a(s)-b_{as}(s) \big) \, \d(\Lambda_{as} -\Lambda^*_{as})(s).
	\end{split}\end{align}
	The latter decomposition is invariant with respect to a reordering of the components of $X$, whereas the SU decomposition changes.
	Therefore, we recommend to replace the traditional SU  decomposition \eqref{SUGermanLI}  by the  ISU  decomposition \eqref{ISUGermanLI}.
\end{example}

\begin{example}
	We choose the setting from Example \ref{ExpRX3} but adopt the mean portfolio perspective.
	By applying Theorem \ref{Theorem:ISU2} and Proposition \ref{AggregationOfISUs} we obtain the ISU decomposition
	\begin{align*}
		D_u(t)&=0,\\
		D_{j}(t)&=\int_{(0,t]}\frac{1}{\kappa(s-)} p_{aj}(0,t-) \Big( V_j^*(s-)\d (\widetilde{\Phi}-\Phi^*)(s) -\sum_{k:k \neq j} R_{jk}^*(s)\d (\Lambda_{jk}-\Lambda_{jk}^*)(s) \Big), \quad j \in \mathcal{Z}.
	\end{align*}
	As a special case this ISU decomposition includes heuristic approaches of Ramlau-Hansen (1991, formula (3.2)) and Norberg (1999, formula (5.7)) for subdividing mean portfolio surplus in a state-wise manner.
\end{example}

\newpage
\section{Alternative decomposition principles}
In section 4, we already mentioned that the ISU decomposition may depend on the update order. In this section, we want to elaborate on that point by discussing two alternative decomposition principles in the setup of chapter 4. Instead of updating the sources of risk sequentially, we could also update only one source of risk at a time and quantify its impact on total revaluation surplus $R(t)-R(0)$. More precise, for any  partition $\mathcal{T}(t) = \{0=t_0 < t_1 < \cdots < t_k =t \}$ of the interval $[0,t]$ we can decompose
\begin{align*}
 R(t)-R(0)  &= U(t, \ldots, t) -U(0, \ldots,0)\\
  &= \sum_{l=0}^{k-1} \Big( U(t_{l+1},t_{l}, \ldots,t_{l})-U(t_{l}, \ldots,t_{l})\Big) \\
   &\quad +  \sum_{l=0}^{k-1} \Big( U(t_{l},t_{l+1}, t_{l},\ldots,t_{l})-U(t_{l}, \ldots,t_{l})\Big) \\
  &\quad + \ldots \\
  & \quad + \sum_{l=0}^{k-1}  \Big( U(t_{l},\ldots, t_{l}, t_{l+1})-U(t_{l},\ldots, , t_{l})\Big) \\
  &\quad +  \sum_{l=0}^{k-1} \Big( U(t_{l+1},\ldots,t_{l+1})-U(t_{l},\ldots,t_{l})\Big)\\
  &\quad -\sum_{l=0}^{k-1} \Big(U(t_{l+1},t_l,\ldots,t_l)-U(t_{l},\ldots,t_{l})+\ldots+U(t_{l},\ldots,t_l,t_{l+1})-U(t_{l},\ldots,t_{l})\Big)
 \end{align*}
Here, the first $m$ sums quantify the single effect of the corresponding source of risk. Following Biewen (2014), we call them the \textit{ceteris paribus effects}. Since the ceteris paribus effects do not necessarily add up to the total revaluation surplus $R(t)-R(0)$, we get an extra term in the last two lines, which is called the \textit{interaction effect} (cf. Biewen (2014)). Based on this construction, we get a decomposition principle with a joint risk factor.

\begin{definition}
The  random vector   $D(t)=(D_1(t), \ldots, D_m(t),\overline{D}(t))$  defined by
\begin{align}\begin{split}
  D_1(t) &= \sum_{l=0}^{k-1} \Big( U(t_{l+1},t_{l}, \ldots,t_{l})-U(t_{l},t_{l}, \ldots,t_{l})\Big), \\
  & \cdots  \\
    D_{m}(t) &= \sum_{l=0}^{k-1 }  \Big( U(t_{l},\ldots, t_{l}, t_{l+1})-U(t_{l},\ldots, t_{l})\Big), \\
    \overline{D}(t) &=R(t)-R(0)-\sum\limits_{j=1}^m D_j(t)
 \end{split}\end{align}
  is called  the \emph{OAT (one-at-a-time) decomposition}  of $R(t)-R(0)$ with respect to $\mathcal{T}(t)$.
\end{definition}
The OAT decomposition principle is also known in economics, see for example Biewen (2014). In contrast to the ISU decomposition, the OAT decomposition is symmetric with respect to the risk factors, i.e. it does not depend on the order of the risk basis. Nevertheless, we get a joint risk factor that cannot be assigned to any source of risk. In Chapter 4, we faced the order dependence of the SU decomposition by considering increasing sequences of partitions of $[0,t]$. Similarly, we face the unassignable interaction effect in the OAT decomposition.

Let  $\mathcal{T}_n(t)=\{ 0=t^n_0 < t^n_1 < \cdots  < t^n_{k_n} =t \}$, $n \in \mathbb{N}$, be a sequence of
 partitions of $[0,t]$  with vanishing step lengths (i.e.~$  \lim_{ n\rightarrow \infty} \max_{1 \leq l  \leq k_{n}} |t^n_{l}-t^n_{l-1}| =0$).
 For each $n \in \mathbb{N}$ let  $D^n(t)=(D_1^{n}(t), \ldots ,D_{m}^{n}(t),\overline{D}^n(t))$ be the OAT decomposition of $R(t)-R(0)$ with respect to $\mathcal{T}_n(t)$. We are looking for a random vector  $D(t)=(D_1(t),\ldots,D_m(t),\overline{D}(t))$ that satisfies
\begin{align}\label{IOATdecomposition}\begin{split}
 D_i(t) &= \plim_{n \rightarrow \infty} D^{n}_i(t), \quad i \in \{1, \ldots, m\}, \\
 \overline{D}(t)&=\plim_{n \rightarrow \infty} \overline{D}^n(t).
 \end{split}\end{align}
\begin{definition}
 Let $(\mathcal{T}_n(t))_{n \in \mathbb{N}}$ be a sequence of partitions of $[0,t]$ with vanishing step lengths.
 If  $D(t)=(D_1(t),\ldots,D_m(t),\overline{D}(t))$ satisfies \eqref{IOATdecomposition}, then we call $D(t)$ the \emph{IOAT (infinitesimal one-at-a-time) decomposition}  of $R(t)-R(0)$ with respect to $(\mathcal{T}_n(t))_{n \in \mathbb{N}}$.
\end{definition}
The next theorem characterizes the relation between the ISU decomposition and the IOAT decomposition.
\begin{theorem}\label{Theorem:ISU=IOAT}
The following statements are equivalent:
\begin{itemize}
\item[a)] The ISU decomposition is independent of update order.
\item[b)] For each update order, the ISU decomposition is equal to the ceteris paribus effects of the IOAT decomposition.
\end{itemize}
In both cases, the interaction effect is zero.
\end{theorem}
\begin{proof}
The proof follows Biewen (2014). Let us fix a source of risk ($i=1,\ldots,m$).
Choosing an update order, such that the this source of risk is updated firstly, the
corresponding risk factor of the ISU decomposition coincides per definition with the ceteris paribus effect of the IOAT decomposition. If the ISU decomposition is independent of update order, the risk factor, corresponding to the fixed source of risk, equals the ceteris paribus effect of the IOAT decomposition for each update order. \par
 On the other hand, the statement in b) directly implies, that the ISU decomposition is independent of update order. Furthermore, if the ISU decomposition equals the IOAT decomposition, then the ceteris paribus effects sum up to total risk $R(t)-R(0)$, therefore the interaction effect is zero.
\end{proof}
By subdividing the interaction effect into different groups of interaction effects (depending on the number of involved risk factors), Biewen (2014) even shows that the particular interaction effects are zero if and only if the ISU decomposition is independent of update order. \par
If the interaction effect is non-zero, neither the ISU decomposition nor the IOAT decomposition yields a unique decomposition in the sense of (\ref{aim:decomp}). One possible solution for this problem is to build a decomposition principle based on the ISU decomposition principle that is symmetric with respect to the sources of risk. For that, let $\pi\colon \{1,\ldots,m\} \to \{1,\ldots,m\}$ be a permutation that represents an update order for the ISU decomposition. The set of all possible permutations on $\{1,\ldots,m\}$ is denoted by $\sigma_m$.
\begin{definition}
Let $(\mathcal{T}_n(t))_{n \in \mathbb{N}}$ be an increasing sequence of partitions of $[0,t]$ with vanishing step lengths and let $\pi\in \sigma_m$. Further, let $D^\pi(t)=(D^\pi_1(t),\ldots,D^\pi_m(t))$ denote the ISU decomposition of $R(t)-R(0)$ with respect to $\pi$ and with respect to $(\mathcal{T}_n)_n$. The  random vector   $D(t)=(D_1(t), \ldots, D_m(t))$  defined by
\begin{align}\begin{split}
  D_{1}(t) &= \frac{1}{m!}\sum\limits_{\pi\in \sigma_m} D^{\pi}_{\pi(1)}(t), \\
  & \cdots  \\
    D_{m}(t) &= \frac{1}{m!}\sum\limits_{\pi\in \sigma_m} D^{\pi}_{\pi(m)}(t),
 \end{split}\end{align}
  is called  the \emph{averaged ISU decomposition}  of $R(t)-R(0)$ with respect to $(\mathcal{T}_n(t))_{n \in \mathbb{N}}$.
\end{definition}
In a similar manner, Shorrocks (2013) proposes the averaged SU decomposition (without taking limits) in economics literature. By construction, the averaged ISU decomposition principle is symmetric with respect to the risk basis and therefore gives a unique surplus decomposition even if the interaction effect is non-zero. Furthermore, the averaged ISU decomposition is in line with the previously proposed decomposition principles as the next theorem shows.
\begin{theorem}
If the ISU decomposition is independent of update order, then ISU (for each update order), IOAT and averaged ISU yield the same decomposition.
\end{theorem}
\begin{proof}
Assume that the ISU decomposition principle yields a decomposition $(D_1(t),\ldots,D_m(t))$ for each update order. Then, by Theorem \ref{Theorem:ISU=IOAT}, the ISU decomposition is equal to the IOAT decomposition for each update order. Furthermore, it holds $D^{\pi}_{\pi(i)}(t)=D_i(t)$, $i=1,\ldots,m$ for every permutation $\pi$. Since $\# \sigma_m=m!$, the averaged ISU decomposition is also given by $(D_1(t),\ldots,D_m(t))$.
\end{proof}
As shown in Chapter 5, the ISU decompositions in our life insurance model do not depend on the update order. Thus, we directly get the following result.
\begin{corollary}
For all examples in section \ref{sectionExamples}  the IOAT decomposition and the averaged ISU decomposition are both equal to the ISU decomposition.
\end{corollary}

\section{Conclusion}

The ISU decomposition principle allows us to unite the various surplus decomposition formulas from the life insurance literature under one banner. By doing so, we replace the common heuristic constructions by a general and consistent decomposition principle.  Furthermore, the generality of the ISU construction paves the way for future extensions of life insurance bonus schemes. This is in particular relevant for new insurance forms that reward risk averse behaviour by activity tracking. The activity of an insured is a risk factor that contributes to the overall surplus. An attribution of profits and losses to the tracked activities of the insured is the necessary prerequisite for an activity bonus or activity malus. The ISU decomposition principle provides a powerful tool for that.  As ISU decompositions can be naturally approximated by SU decompositions, they are actually easy to implement in insurance practice.

The ISU concept is also useful beyond surplus decompositions in with-profit life insurance. Whenever the profits and losses of a financial entity shall be decomposed with respect to their sources, the ISU decomposition is a viable and convincing tool for that. One of the key ideas in this paper is to overcome the well-known limitations of SU and OAT decompositions by an infinitesimal approach.
So our ISU, averaged ISU and IOAT concepts may be of help in all that applications where SU and OAT decompositions are currently in use. We recommend to generally use (averaged) ISU rather than IOAT decompositions, since the former always guarantee additivity of the decompositions.

\section*{References}

\bigskip {\small
\begin{list}{}{\leftmargin1cm\itemindent-1cm\itemsep0cm}
\item{Asmussen, S.~and Steffensen, M., 2020. Risk and Insurance: A Graduate Text. Vol. 96. Springer Nature.}

\item{Biewen, M., 2014. A general decomposition formula with interaction effects. Applied Economics Letters, 21(9), 636-642.}

\item{Fortin, N., Lemieux, T., Firpo, S., 2011. Chapter 1-decomposition methods in economics. Volume 4, Part A of Handbook of Labor Economics. Elsevier 10, S0169-7218.}

\item {Milbrodt, H. and Helbig, M., 1999. Mathematische Methoden der Personenversicherung. Walter de Gruyter.}

\item{ M{\o}ller, T., Steffensen, M., 2007. Market-valuation methods in life and pension insurance. Cambridge University Press.}

\item{Norberg, R., 1999. A theory of bonus in life insurance. Finance and Stochastics 3/4, 373-390.}

\item{Protter, P.E., 2005. Stochastic Integration and Differential equations, 2nd edition, Springer.}

\item{Ramlau-Hansen, H., 1988. The emergence of profit in life insurance. Insurance: Mathematics and Economics, 7(4), 225-236.}

\item{Ramlau-Hansen, H., 1991. Distribution of surplus in life insurance. ASTIN Bulletin: The Journal of the IAA, 21/1, 57-71.}

\item{Shorrocks, A.F., 2013. Decomposition procedures for distributional analysis: a unified framework based on the Shapley value. The Journal of Economic Inequality 11, 99–126.}

\item{Werner, D., 2018. Funktionalanalysis, 8th edition, Springer.}
\end{list}}

\appendix
\section{Appendix}
\subsection{Proofs}
\begin{proof}[Proof of Lemma \ref{lemmasurplusdecomp}]
As a shorthand notation, we define multivariate processes $C^*=(C_1^*,\ldots,C_n^*)^\top$ and $\overline{C}=(\overline{C}_1,\ldots,\overline{C}_n)^{\top}$ by \begin{align*}
\d C^*_j(u)&=\d B_j(u)+\sum\limits_{k:k\neq j} b_{jk}(u)\d \Lambda^*_{jk}(u),  \ C^*_j(0)=0,\\
\d \overline{C}_j(u)&=\d B_j(u)+\sum\limits_{k:k\neq j} b_{jk}(u)\d \overline{\Lambda}_{jk}(u),  \ \overline{C}_j(0)=0.
\end{align*}
Note that $C^*$ and $\overline{C}$ are column vectors.
The vectorial process $I=(I_j)_j$ shall combine all state processes as a row vector.
We further define $$W(u)\coloneqq -H((\Phi^*, \Lambda^*)+(\overline{\Phi}-\Phi^*,\overline{\Lambda}-\Lambda^*)^u).$$ Then for $u\in (0,t]$, we get
\begin{align*}
W(u)&=-\int_{[0,u]} \frac{1}{\overline{\kappa}(s)} I(0)\overline{p}(0,s-) \d \overline{C}(s) - \frac{1}{\overline{\kappa}(u)} \int_{(u,T]} \frac{\kappa^*(u)}{\kappa^*(s)}  I(0) \overline{p}(0,u)p^*(u,s-)\d C^*(s)\\
&=-\int_{[0,u]} \frac{1}{\overline{\kappa}(s)} I(0)\overline{p}(0,s-) \d \overline{C}(s) - \frac{\kappa^*(u)}{\overline{\kappa}(u)} I(0) \overline{p}(0,u)q^*(0,u)\, Y(u) ,
\end{align*}
for $Y(u)= \int_{(u,T]} \frac{1}{\kappa^*(s)}  p^*(0,s-)\d C^*(s)$. Analogously to 
 $\Lambda^*_M$, let $\overline{\Lambda}_M $ denote the matrix-valued process $\overline{\Lambda}_M = (\overline{\Lambda}_{jk})_{jk}$ with $\overline{\Lambda}_{jj}\coloneqq -\sum_{k:k\neq j} \overline{\Lambda}_{jk}$.
By applying It\^{o}'s formula and using the assumption that $(\Phi^*,\overline{\Phi})$ and $(\Lambda^*,\overline{\Lambda},(B_j)_j)$ have no common jumps, we can show that
\begin{align*}
\d W(u)
&= - \frac{1}{\overline{\kappa}(u)} I(0)\overline{p}(0,u-) \d (\overline{C}-C^*)(u) \\
&\ \ \ \ -I(0)\frac{\kappa^*(u-)}{\overline{\kappa}(u-)} \overline{p}(0,u-)q^*(0,u-) Y(u-) \d (\Phi^*-\widetilde{\overline{\Phi}}-[\Phi^*,\widetilde{\overline{\Phi}}])(u)\\
&\ \ \ \ -I(0)\frac{\kappa^*(u-)}{\overline{\kappa}(u-)} \overline{p}(0,u-)\d (\overline{\Lambda}_M-\Lambda^*_M)(u)\, q^*(0,u)  Y(u)\\
&= - \frac{1}{\overline{\kappa}(u)} I(0)\overline{p}(0,u-) \d (\overline{C}-C^*)(u) \\
&\ \ \ \ -\frac{1}{\overline{\kappa}(u-)} I(0)\overline{p}(0,u-) \left(\int_{[u,T]}  \frac{\kappa^*(u-)}{\kappa^*(s)}  p^*(u-,s-)\d C^*(s)\right)\d (\Phi^*-\widetilde{\overline{\Phi}}-[\Phi^*,\widetilde{\overline{\Phi}}])(u)\\
&\ \ \ \ -\frac{1}{\overline{\kappa}(u)}I(0) \overline{p}(0,u-)\d (\overline{\Lambda}_M-\Lambda^*_M)(u) \left(\int_{(u,T]}  \frac{\kappa^*(u)}{\kappa^*(s)}  p^*(u,s-)\d C^*(s)\right)
\end{align*}
where we used Lemma \ref{itoinverse2} to get
\begin{align*}
&\d \left( \frac{\kappa^*(u)}{\overline{\kappa}(u)} \overline{p}(0,u)q^*(0,u)\right)\\
&= \frac{\kappa^*(u-)}{\overline{\kappa}(u-)} \d\left(\overline{p}(0,u)q^*(0,u)\right)+\overline{p}(0,u-)q^*(0,u-)\d \left(\frac{\kappa^*(u)}{\overline{\kappa}(u)} \right)+\d \left[\frac{\kappa^*}{\overline{\kappa}}, \overline{p}(0,\cdot)q^*(0,\cdot) \right](t)\\
&=\frac{\kappa^*(u-)}{\overline{\kappa}(u-)} \overline{p}(0,u-)\d (\overline{\Lambda}_M-\Lambda^*_M)(u)\, q^*(0,u)+\frac{\kappa^*(u-)}{\overline{\kappa}(u-)} \overline{p}(0,u-)q^*(0,u-)\d (\Phi^* - \widetilde{\overline{\Phi}}-[\Phi^*,\widetilde{\overline{\Phi}}])(u)
\end{align*}
with $\widetilde{\overline{\Phi}}(u)=\overline{\Phi}(u)-[\overline{\Phi},\overline{\Phi}]^c(u)-\sum_{0<s\leq u} (1+\Delta \overline{\Phi}(s))^{-1} (\Delta \overline{\Phi}(s))^2$.
Component-wise evaluation and integration on $(0,t]$ gives us the assertion.
\end{proof}

\begin{proof}[Proof of Theorem \ref{Theorem:ISU1}]
Let $J_{\Phi}\subseteq \mathcal{Z}$ and $J_u, J_s \subseteq \mathcal{A}\coloneqq \{(j,k)\in \mathcal{Z}^2:j\neq k\}$. For $s\leq t$, we define
\begin{align*}
X_{\Phi,J_\Phi,j}^{s,t}\coloneqq \begin{cases} X_{\Phi,j}^{s},\ j\notin J_{\Phi}, \\ X_{\Phi,j}^{t},\ j\in J_{\Phi}, \end{cases}
\end{align*}
as well as
\begin{align*}
X_{u,J_u,jk}^{s,t}\coloneqq \begin{cases} X_{u,jk}^{s},\ (j,k)\notin J_u, \\ X_{u,jk}^{t},\ (j,k)\in J_u, \end{cases}
X_{s,J_s,jk}^{s,t}\coloneqq \begin{cases} X_{s,jk}^{s},\ (j,k)\notin J_s, \\ X_{s,jk}^{t},\ (j,k)\in J_s. \end{cases}
\end{align*}
We further set $$X_{\Phi,J_{\Phi}}\coloneqq \sum_{j} X^{0,T}_{\Phi,J_{\Phi},j},\; X_{u,J_u}\coloneqq \big( X^{0,T}_{u,J_u,jk}\big)_{jk},\; X_{s,J_s}\coloneqq \big( X^{0,T}_{s,J_s,jk}\big)_{jk},$$
 where $X^{0,T}_{u,J_u,jj}=-\sum_{k:k\neq j} X^{0,T}_{u,J_u,jk}$ and $X^{0,T}_{s,J_s,jj}=-\sum_{k:k\neq j} X^{0,T}_{s,J_s,jk}$.
Let $\Phi^{J_\Phi}\coloneqq X_{\Phi,J_{\Phi}}+\Phi^*$ and let $\kappa^{J_\Phi}$ denote the solution of $\d \kappa^{J_\Phi}(t)= \kappa^{J_\Phi}(t-)\d \Phi^{J_\Phi}(t)$ with $\kappa^{J_\Phi}(0)=1$. Similarly, for $J=(J_u,J_s)$ let $\Lambda^{J}\coloneqq X_{u,J_u}+X_{s,J_s}+\Lambda_M^*$ and let $p^J=(p_{jk})_{j,k}$ denote the solution of $p^J(s,\d t)=p^J(s,t-)\d \Lambda^J(t)$ with $p^J(s,s)$ being the identity matrix.\par
Let $t\in [0,T]$ and let $(\mathcal{T}_n(t))_n$ be a sequence of partitions of $[0,t]$. For a simpler notation, we only write $t_k$ instead of $t_k^n$ for the grid points in $\mathcal{T}_n$. Throughout the proof, let $\alpha_n(s)$ be the left point of $s$ in $\mathcal{T}_n(t)$, i.e. $\alpha_n(s)\coloneqq t_k$ if $s\in (t_k,t_{k+1}]$. For notational convenience, we write
$$ \varrho_{J_{\Phi},J_u,J_s}^{t_k,t_{k+1}}=\varrho((X_{\Phi,J_{\Phi},j}^{t_k, t_{k+1}})_j, (X_{u,J_u,jk}^{t_k,t_{k+1}})_{jk:j\neq k},(X_{s,J_s,jk}^{t_k,t_{k+1}})_{jk:j\neq k}).$$
It is sufficient to show that
\begin{itemize}
\item[i)] $ \plim\limits_{n\to \infty} \sum_{t_k,t_{k+1}\in \mathcal{T}_n(t)} \left(\varrho_{J_{\Phi}\cup \{j_0\},J_u,J_s}^{t_k,t_{k+1}}-\varrho_{J_{\Phi},J_u,J_s}^{t_k,t_{k+1}}\right)= D_{\Phi,j_0}(t)$, $j_0\in \mathcal{Z}\setminus J_{\Phi}$,
\item[ii)] $\plim\limits_{n\to \infty} \sum_{t_k,t_{k+1}\in \mathcal{T}_n(t)} \left(\varrho_{J_{\Phi},J_u\cup {\{(j_0,k_0)\}}, J_s}^{t_k,t_{k+1}}-\varrho_{J_\Phi,J_u, J_s}^{t_k,t_{k+1}}\right)=D_{u,j_0k_0}(t)$, $(j_0,k_0)\in \mathcal{A}\setminus J_{u}$,
\item[iii)] $\plim\limits_{n\to \infty} \sum_{t_k,t_{k+1}\in \mathcal{T}_n(t)} \left(\varrho_{J_{\Phi},J_u, J_s\cup {\{(j_0,k_0)\}}}^{t_k,t_{k+1}}-\varrho_{J_{\Phi},J_u, J_s}^{t_k,t_{k+1}}\right)= D_{s,j_0k_0}(t)$, $(j_0,k_0)\in \mathcal{A}\setminus J_{s}$.
\end{itemize} We prove the convergences consecutively.

\begin{itemize}
\item[i)] Let $\overline{J}_{\Phi}=J_{\Phi}\cup \{j_0\}$, $j_0\in \mathcal{Z}\setminus J_{\Phi}$ and let $$\Delta(u,s)=\frac{\kappa^{\overline{J}_{\Phi}}(u)}{\kappa^{\overline{J}_{\Phi}}(s)}-\frac{\kappa^{J_{\Phi}}(u)}{\kappa^{J_{\Phi}}(s)},\; u\leq s.$$ We define stochastic processes
\begin{align*}
\xi_{\Phi,j_0,n}(s)&=\frac{1}{\kappa(\alpha_n(s))}\frac{\kappa^{\overline{J}_{\Phi}}(\alpha_n(s))}{\kappa^{\overline{J}_{\Phi}}(s-)} \sum\limits_{g} I_g(t_k)\sum\limits_{j} p^{J}_{gj}(\alpha_n(s),s-)V_j^*(s-)I_{j_0}(s-),\\
\xi_{\Phi,j,n}(s)&= \frac{\Delta(\alpha_n(s),s-)}{\kappa(\alpha_n(s))} \sum\limits_{g} I_g(\alpha_n(s))\sum\limits_{j} p^{J}_{gj}(\alpha_n(s),s-)V_j^*(s-)I_{j}(s-), \ j\in J_{\Phi},\\
\xi_{us,jk,n}(s)&=-\sum\limits_{g} I_g(\alpha_n(s)) \frac{\Delta(\alpha_n(s),s)}{\kappa(\alpha_n(s))} p^{J}_{gj}(\alpha_n(s),s-)R_{jk}^*(s),\ (j,k)\in J_u\cup J_s,
\end{align*}
where $s\in [0,t]$. With Lemma \ref{lemmasurplusdecomp}, we have
\begin{align*}
&\sum_{t_k,t_{k+1}\in \mathcal{T}_n(t)} (\varrho_{J_{\Phi}\cup \{j_0\},J_u,J_s}^{t_k,t_{k+1}}-\varrho_{J_{\Phi},J_u,J_s}^{t_k,t_{k+1}})\\
&=\sum_{t_k,t_{k+1}\in \mathcal{T}_n(t)} (\varrho_{J_{\Phi}\cup \{j_0\},J_u,J_s}^{t_k,t_{k+1}}-\varrho_{\O,\O,\O}^{t_k,t_{k+1}}-(\varrho_{J_{\Phi},J_u,J_s}^{t_k,t_{k+1}}-\varrho_{\O,\O,\O}^{t_k,t_{k+1}}))\\
&= \sum\limits_{j\in \overline{J}_\Phi}\int_{(0,t]}\xi_{\Phi,j,n}(s)\d (\widetilde{\Phi}-\Phi^*)(s)+\sum\limits_{(j,k)\in J_u} \int_{(0,t]}\xi_{us,jk,n}(s)\d (N_{jk}-\Lambda_{jk})(s) \\
&\quad \ +\sum\limits_{(j,k)\in J_s} \int_{(0,t]}\xi_{us,jk,n}(s)\d (\Lambda_{jk}-\Lambda^*_{jk})(s)
\end{align*}
where $\widetilde{\Phi}(s)=\Phi(s)-[\Phi,\Phi]^c(s)-\sum_{0<u\leq s} (1+\Delta \Phi^1(u))^{-1} (\Delta \Phi(u))^2$. Here, we used that
\begin{align*}
\d (\widetilde{\Phi^{J_\Phi}}-\Phi^{*}+[\widetilde{\Phi^{J_\Phi}},\Phi^*])(s)=\sum\limits_{j\in J_\Phi}I_j(s-)\d (\widetilde{\Phi}-\Phi^*)(s).
\end{align*}
Since for every $s\in [0,t]$ we almost surely have
\begin{align*}
\lim\limits_{n\to\infty} \xi_{\Phi,j_0,n}(s)&=\frac{1}{\kappa(s-)}I_{j_0}(s-)V_{j_0}^*(s-),\\
\lim\limits_{n\to\infty} \xi_{\Phi,j,n}(s)&=0,  \; j \in J_{\Phi},\\
\lim\limits_{n\to\infty} \xi_{us,jk,n}(s)&= I_j(s-) \frac{\Delta(s-,s)}{\kappa(s-)}R_{jk}^*(s),\; (j,k)\in J_u\cup J_s,
\end{align*}
and since $\Delta(s-,s)\d (N_{jk}-\Lambda_{jk})(s)=\Delta(s-,s)\d (\Lambda_{jk}-\Lambda^*_{jk})(s)=0$ almost surely,
the Dominated Convergence Theorem for stochastic integrals (cf. Protter, 2005, Chapter IV, Theorem 32) yields
\begin{align*}
\plim\limits_{n\to \infty} \sum_{t_k,t_{k+1}\in \mathcal{T}_n(t)} \left(\varrho_{J_{\Phi}\cup \{j_0\},J_u,J_s}^{t_k,t_{k+1}}-\varrho_{J_{\Phi},J_u,J_s}^{t_k,t_{k+1}}\right)=D_{\Phi,j_0}(t).
\end{align*}

\item[ii)]
Let $\overline{J}=(J_u\cup \{j_0,k_0\},J_s)$, $(j_0,k_0)\in \mathcal{A}\setminus  J_{u}$ and let
$$\Delta_{jk}(u,s)\coloneqq p^{\overline{J}}_{jk}(u,s)-p^{J}_{jk}(u,s),\; u\leq s .$$
We define stochastic processes
\begin{align*}
\xi_{\Phi,j,n}(s)&=\frac{1}{\kappa(\alpha_n(s))}\frac{\kappa^{J_{\Phi}}(\alpha_n(s))}{\kappa^{J_{\Phi}}(s-)} \sum\limits_{g} I_g(\alpha_n(s))\sum\limits_{j} \Delta_{gj}(\alpha_n(s),s-)V_j^*(s-)I_{j}(s-),\; j\in J_\Phi,\\
\xi_{us,jk,n}(s)&= - \sum\limits_{g} I_g(\alpha_n(s)) \frac{1}{\kappa(\alpha_n(s))} \frac{\kappa^{J_\Phi}(\alpha_n(s))}{\kappa^{J_\Phi}(s)} \Delta_{gj}(\alpha_n(s),s-)R_{jk}^*(s),\; (j,k)\in J_u\cup J_s,\\
\xi_{u,j_0k_0,n}(s)&= - \sum\limits_{g} I_g(\alpha_n(s)) \frac{1}{\kappa(\alpha_n(s))} \frac{\kappa^{J_\Phi}(\alpha_n(s))}{\kappa^{J_\Phi}(s)} p^{\overline{J}}_{gj_0}(\alpha_n(s),s-)R_{j_0k_0}^*(s),
\end{align*}
where $s\in [0,t]$. Again with Lemma \ref{lemmasurplusdecomp}, we have
\begin{align*}
&\sum\limits_{t_k,t_{k+1}\in \mathcal{T}_n(t)} (\varrho_{J_{\Phi},J_u\cup {\{(j_0,k_0)\}}, J_s}^{t_k,t_{k+1}}-\varrho_{J_\Phi,J_u, J_s}^{t_k,t_{k+1}}) \\
&=\sum\limits_{t_k,t_{k+1}\in \mathcal{T}_n(t)} (\varrho_{J_{\Phi},J_u\cup {\{(j_0,k_0)\}}, J_s}^{t_k,t_{k+1}}-\varrho_{\O,\O,\O}^{t_k,t_{k+1}}-(\varrho_{J_{\Phi},J_u,J_s}^{t_k,t_{k+1}}-\varrho_{\O,\O,\O}^{t_k,t_{k+1}}))\\
&= \sum\limits_{j\in J_\Phi}\int_{(0,t]}\xi_{\Phi,j,n}(s)\d (\widetilde{\Phi}-\Phi^*)(s)+\sum\limits_{(j,k)\in J_u} \int_{(0,t]}\xi_{us,jk,n}(s)\d (N_{jk}-\Lambda_{jk})(s) \\
&\quad \ +\int_{(0,t]}\xi_{u,j_0k_0,n}(s)\d (N_{j_0k_0}-\Lambda_{j_0k_0})(s) +\sum\limits_{(j,k)\in J_s} \int_{(0,t]}\xi_{us,jk,n}(s)\d (\Lambda_{jk}-\Lambda^*_{jk})(s)
\end{align*}
Since for every $s\in [0,t]$ we almost surely have
\begin{align*}
\lim\limits_{n\to \infty} \xi_{u,j_0k_0,n}(s)&=-\frac{1}{\kappa(s-)}  I_{j_0}(s-)  \frac{\kappa^{J_\Phi}(s-)}{\kappa^{J_\Phi}(s)} R_{j_0k_0}^*(s)  \\
&= -\frac{1}{\kappa(s)}  I_{j_0}(s-)  \frac{1+\Delta \Phi(s)}{1+\Delta \Phi^{J_\Phi}(s)} R_{j_0k_0}^*(s),\\
\lim\limits_{n\to\infty} \xi_{\Phi,j,n}(s)&=\lim\limits_{n\to \infty} \xi_{us,jk,n}(s)=0,\; j\in J_\Phi,\; (j,k)\in J_u\cup J_s,
\end{align*}
and since $\frac{1+\Delta \Phi(s)}{1+\Delta \Phi^{J_\Phi}(s)}\d(N_{j_0k_0}-\Lambda_{j_0k_0})(s)=\d(N_{j_0k_0}-\Lambda_{j_0k_0})(s)$ almost surely, the Dominated Convergence Theorem for stochastic integrals (cf. Protter, 2005, Chapter IV, Theorem 32) yields
\begin{align*}
\plim\limits_{n\to\infty}  \sum\limits_{t_k,t_{k+1}\in \mathcal{T}_n(t)} \left(\varrho_{J_{\Phi},J_u\cup {\{(j_0,k_0)\}}, J_s}^{t_k,t_{k+1}}-\varrho_{J_\Phi,J_u, J_s}^{t_k,t_{k+1}}\right)=D_{u,j_0k_0}(t).
\end{align*}

\item[iii)]
Let $\overline{J}=(J_u,J_s\cup \{j_0,k_0\})$, $(j_0,k_0)\in \mathcal{A}\setminus  J_{s}$ and let
$$\Delta_{jk}(u,s)\coloneqq p^{\overline{J}}_{jk}(u,s)-p^{J}_{jk}(u,s),\; u\leq s.$$
We define stochastic processes
\begin{align*}
	\xi_{\Phi,j,n}(s)&=\frac{1}{\kappa(\alpha_n(s))}\frac{\kappa^{J_{\Phi}}(\alpha_n(s))}{\kappa^{J_{\Phi}}(s-)} \sum\limits_{g} I_g(\alpha_n(s))\sum\limits_{j} \Delta_{gj}(\alpha_n(s),s-)V_j^*(s-)I_{j}(s-),\; j\in J_\Phi,\\
	\xi_{us,jk,n}(s)&= - \sum\limits_{g} I_g(\alpha_n(s)) \frac{1}{\kappa(\alpha_n(s))} \frac{\kappa^{J_\Phi}(\alpha_n(s))}{\kappa^{J_\Phi}(s)} \Delta_{gj}(\alpha_n(s),s-)R_{jk}^*(s),\; (j,k)\in J_u\cup J_s,\\
	\xi_{s,j_0k_0,n}(s)&= - \sum\limits_{g} I_g(\alpha_n(s)) \frac{1}{\kappa(\alpha_n(s))} \frac{\kappa^{J_\Phi}(\alpha_n(s))}{\kappa^{J_\Phi}(s)} p^{\overline{J}}_{gj_0}(\alpha_n(s),s-)R_{j_0k_0}^*(s),
\end{align*}
where $s\in [0,t]$. Again with Lemma \ref{lemmasurplusdecomp}, we have
\begin{align*}
	&\sum\limits_{t_k,t_{k+1}\in \mathcal{T}_n(t)} (\varrho_{J_{\Phi},J_u, J_s\cup {\{(j_0,k_0)\}}}^{t_k,t_{k+1}}-\varrho_{J_\Phi,J_u, J_s}^{t_k,t_{k+1}}) \\
	&=\sum\limits_{t_k,t_{k+1}\in \mathcal{T}_n(t)} (\varrho_{J_{\Phi},J_u, J_s\cup {\{(j_0,k_0)\}}}^{t_k,t_{k+1}}-\varrho_{\O,\O,\O}^{t_k,t_{k+1}}-(\varrho_{J_{\Phi},J_u,J_s}^{t_k,t_{k+1}}-\varrho_{\O,\O,\O}^{t_k,t_{k+1}}))\\
	&= \sum\limits_{j\in J_\Phi}\int_{(0,t]}\xi_{\Phi,j,n}(s)\d (\widetilde{\Phi}-\Phi^*)(s)+\sum\limits_{(j,k)\in J_u} \int_{(0,t]}\xi_{us,jk,n}(s)\d (N_{jk}-\Lambda_{jk})(s) \\
	&\quad \ +\sum\limits_{(j,k)\in J_s} \int_{(0,t]}\xi_{us,jk,n}(s)\d (\Lambda_{jk}-\Lambda^*_{jk})(s)+\int_{(0,t]}\xi_{s,j_0k_0,n}(s)\d (\Lambda_{j_0k_0}-\Lambda^*_{j_0k_0})(s).
\end{align*}
Since for every $s\in [0,t]$ we almost surely have
\begin{align*}
	\lim\limits_{n\to \infty} \xi_{s,j_0k_0,n}(s)&=-\frac{1}{\kappa(s-)}  I_{j_0}(s-)  \frac{\kappa^{J_\Phi}(s-)}{\kappa^{J_\Phi}(s)} R_{j_0k_0}^*(s)  \\
	&= -\frac{1}{\kappa(s)}  I_{j_0}(s-)  \frac{1+\Delta \Phi(s)}{1+\Delta \Phi^{J_\Phi}(s)} R_{j_0k_0}^*(s),\\
	\lim\limits_{n\to\infty} \xi_{\Phi,j,n}(s)&=\lim\limits_{n\to \infty} \xi_{us,jk,n}(s)=0, \; j\in J_\Phi,\; (j,k)\in J_u\cup J_s,
\end{align*}
and since $\frac{1+\Delta \Phi(s)}{1+\Delta \Phi^{J_\Phi}(s)}\d(\Lambda_{j_0k_0}-\Lambda^*_{j_0k_0})(s)=\d(\Lambda_{j_0k_0}-\Lambda^*_{j_0k_0})(s)$ almost surely, the Dominated Convergence Theorem for stochastic integrals (cf. Protter, 2005, Chapter IV, Theorem 32) yields
\begin{align*}
	\plim\limits_{n\to\infty}  \sum\limits_{t_k,t_{k+1}\in \mathcal{T}_n(t)} \left(\varrho_{J_{\Phi},J_u, J_s\cup {\{(j_0,k_0)\}}}^{t_k,t_{k+1}}-\varrho_{J_\Phi,J_u, J_s}^{t_k,t_{k+1}}\right)=D_{s,j_0k_0}(t).
\end{align*}
\end{itemize}
\end{proof}

\begin{proof}[Proof of Theorem \ref{Theorem:ISU2}] The model framework, introduced in chapter 3, entails that the integrability assumption in Lemma \ref{ISU:ExpectationConsistency} for the SU decomposition (cf. proof of Theorem \ref{Theorem:ISU1}) is satisfied. Thus, applying Lemma \ref{ISU:ExpectationConsistency} with $\mathcal{G}=\sigma(\Phi,\Lambda)$ to the ISU decomposition in Theorem \ref{Theorem:ISU1} and using the martingale property of $\d N_{jk}-I_j(t-) \d \Lambda_{jk}(t)$ with respect to the natural completed filtration of the random vector $(Z^t,\Phi,\Lambda)_{t\geq 0}$ give the desired result.
\end{proof}

\begin{proof}[Proof of the SU decomposition in the time-discrete case (Example \ref{Example:RinGermany})] In the setting of Example \ref{Example:RinGermany}, the functional $H$ in \eqref{functionaldef} takes the form
\begin{align*}
H(\overline{\Phi},\overline{\Lambda}_{ad},\overline{\Lambda}_{as})=\sum\limits_{l=0}^T e^{-\int_0^l \overline{\phi}(u)\d u} \phantom{.}_{l}\overline{p}_{x} b_l+\sum\limits_{l=1}^T e^{-\int_0^l \overline{\phi}(u)\d u} \phantom{.}_{l-1}\overline{p}_{x}(\overline{q}_{x+l} d_l+\overline{r}_{x+l} s_l)
\end{align*}
Furthermore, for the risk basis $X=(\Phi-\Phi^*,\Lambda_{ad}-\Lambda_{ad}^*,\Lambda_{as}-\Lambda_{as}^*)$, the mapping $\varrho$ is given by $\varrho(X^t)=-H((\Phi^*,\Lambda_{ad}^*,\Lambda_{as}^*)+X^t)$.
We prove the three equations consecutively.
\begin{itemize}
\item[i)] We have that \begin{align*}
 &U(k+1,k,k)-U(k,k,k)\\
 &=\varrho(\Phi^{k+1},\Lambda_{ad}^k,\Lambda_{as}^k)-\varrho(\Phi^{k},\Lambda_{ad}^k,\Lambda_{as}^k) \\
&=e^{-\int_0^{k+1} \phi(u)\, \d u } \phantom{.}_{k}p_{x}((1+i_k)V_a^*(k)-p^*_{x+k}\,b_{k+1}-q^*_{x+k}\,d_{k+1}-r^*_{x+k}\,s_{k+1}-p_{x+k}^*\,V_a^*(k+1))
\end{align*}
Since
\begin{align*}
& -p^*_{x+k}\,b_{k+1}-q^*_{x+k}\,d_{k+1} -r^*_{x+k}\,s_{k+1}-p_{x+k}^*\,V_a^*(k+1)= -(1+i_k^*)V_a^*(k),
\end{align*}
we get the first equation.
\item[ii)] With similar calculations as in i) we get
\begin{align*}
&U(k+1,k+1,k)-U(k,k,k)\\
&=\varrho(\Phi^{k+1},\Lambda_{ad}^{k+1},\Lambda_{as}^{k})-\varrho(\Phi^{k},\Lambda_{ad}^{k},\Lambda_{as}^{k}) \\
&=-e^{-\int_0^{k+1} \phi(u)\, \d u } \phantom{.}_{k}p_{x}((1-q_{x+k}-r^*_{x+k})b_{k+1}+ q_{x+k}\,d_{k+1}
+r^*_{x+k}\,s_{k+1})\\
&\ \ \ \ -e^{-\int_0^{k+1} \phi(u)\, \d u } \phantom{.}_{k}p_{x}(1-q_{x+k}-r^*_{x+k})V_a^*(k+1)+ e^{-\int_0^{k} \phi(u)\, \d u }\phantom{.}_{k}p_{x}\,V_a^*(k) \\
&= e^{-\int_0^{k+1} \phi(u)\, \d u } \phantom{.}_{k}p_{x}(V_a(k+1-)-d_{k+1})(q_{x+k}-q^*_{x+k})  +  e^{-\int_0^{k+1} \phi(u)\, \d u } \phantom{.}_kp_{x} \,V^*_a(k)\big(i_k-i_k^*\big)
\end{align*}
The second equality follows then by substracting $U(k+1,k,k)-U(k,k,k)$ (cf. i)) from $U(k+1,k+1,k)-U(k,k,k)$.
\item[iii)] For the third equality, we can use the results from i) and ii) to obtain
\begin{align*}
&U(k+1,k+1,k+1)-U(k+1,k+1,k)\\
&=R(k+1)-R(k)-(U(k+1,k+1,k)-U(k+1,k,k))-(U(k+1,k,k)-U(k,k,k))\\
&=e^{-\int_0^{k+1} \phi(u)\, \d u } \phantom{.}_kp_{x}  \big(V^*_a(k+1-)-s_{k+1} \big)  \,\big(r_{x+k} -r^*_{x+k}\big).
\end{align*}
\end{itemize}
\end{proof}

\subsection{Technical results}
 Analogously to
 $\Lambda^*_M$, let $\overline{\Lambda}_M $ denote the matrix-valued process $\overline{\Lambda}_M = (\overline{\Lambda}_{jk})_{jk}$ with $\overline{\Lambda}_{jj}\coloneqq -\sum_{k:k\neq j} \overline{\Lambda}_{jk}$, and define $\Lambda'_M $ likewise.
 
\begin{lemma}\label{itoinverse}
Let $(\overline{\Phi},\overline{\Lambda})$ be a valuation basis.
\begin{itemize}
\item[a)] Let $\overline{\kappa}$ be the solution of the stochastic differential $\d \overline{\kappa}(t)=\overline{\kappa}(t-)\d \overline{\Phi}(t)$ with $\overline{\kappa}(0)=1$. Then it holds that
    \begin{align*}
\d \left(\frac{1}{\overline{\kappa}(t)}\right)=-\frac{1}{\overline{\kappa}(t-)}\d \widetilde{\overline{\Phi}}(t),
\end{align*}
where $\widetilde{\overline{\Phi}}(t)=\overline{\Phi}(t)-[\overline{\Phi},\overline{\Phi}]^c(t)-\sum_{0<s\leq t} (1+\Delta \overline{\Phi}(s))^{-1} (\Delta \overline{\Phi}(s))^2$.
\item[b)] Let $\overline{p}(s,t)$ be the solution of the matrix-valued stochastic differential equation $\overline{p}(s,\d t)=\overline{p}(s,t-)\d \overline{\Lambda}_M(t)$ with $\overline{p}(s,s)=\mathbb{I}$. Assume that $(\mathbb{I}+\Delta \overline{\Lambda}_M(t))^{-1}$ exists for all $t>0$. Then $\overline{p}(s,t)$ is invertible, and the inverse $\overline{q}(s,t)$ solves the SDE \begin{align*}
\overline{q}(s,\d t)= -(\d G(t))\overline{q}(s,t-)=-(\d \overline{\Lambda}_M(t))\overline{q}(s,t),
\end{align*}
where $G(t)=\overline{\Lambda}_M(t)-\sum_{0<s\leq t} (\Delta\overline{\Lambda}_M(s))^2(I+\Delta \overline{\Lambda}_M(s))^{-1}$.
\end{itemize}
\begin{proof}
\begin{itemize}
\item[a)] Due to the properties of a valuation basis, $\widetilde{\overline{\Phi}}$ is a well-defined semimartingale. Thus, with Theorem V.10.63 of Protter (2005), the assertion follows.
\item[b)] For applying Theorem V.10.63 of Protter (2005) later again, we firstly have to show that $G$ is a well-defined semimartingale. Since $\overline{\Lambda}$ is a càdlàg finite variation process, it suffices to show $(\sum_{0<s\leq t} (\Delta\overline{\Lambda}_M(s))^2(I+\Delta\overline{\Lambda}_M(s))^{-1})_{jk}<\infty$ for all $t>0$ and $j,k$. Let $\Vert\cdot \Vert$, defined by $\Vert A\Vert=n\cdot \max_{j,k} |a_{jk}|$ for a matrix $A=(a_{jk})_{jk}\in \mathbb{R}^{n\times n}$, denote the maximum norm on $\mathbb{R}^{n\times n}$. If $\Vert\Delta \overline{\Lambda}_M(t)\Vert\leq 1/2$, then it holds
$$\Vert (I+\Delta \overline{\Lambda}_M(t))^{-1}\Vert\leq \frac{1}{1-\Vert \Delta \overline{\Lambda}_M(t)\Vert}\leq 2, $$
see for example Werner (2018, Theorem II.1.12). Using this upper bound, the subadditity and the submultiplicity of the norm, we get
\begin{align*}
&\bigg\Vert \sum_{0<s\leq t} (\Delta \overline{\Lambda}_M(s))^2(I+\Delta \overline{\Lambda}_M(s))^{-1}\bigg\Vert \\
&\leq \sum_{\substack{0<s\leq t\\ \Vert \Lambda_M(s)\Vert>1/2}} \bigg\Vert(\Delta \overline{\Lambda}_M(s))^2(I+\Delta \overline{\Lambda}_M(s))^{-1}\bigg\Vert +\sum_{\substack{0<s\leq t\\ \Vert \overline{\Lambda}_M(s)\Vert\leq1/2}}\bigg\Vert (\Delta \overline{\Lambda}_M(s))^2(I+\Delta \overline{\Lambda}_M(s))^{-1}\bigg\Vert \\
&\leq \sum_{\substack{0<s\leq t\\ \Vert \overline{\Lambda}_M(s)\Vert>1/2}} \bigg\Vert(\Delta \overline{\Lambda}_M(s))^2(I+\Delta \overline{\Lambda}_M(s))^{-1}\bigg\Vert +\sum_{\substack{0<s\leq t\\ \Vert \overline{\Lambda}_M(s)\Vert \leq 1/2}}\bigg\Vert \Delta \overline{\Lambda}_M(s)\bigg\Vert.
\end{align*}
The first sum in the latter expression is finite, since $\Vert \Delta\overline{\Lambda}_M (s)\Vert>1/2$ occurs only for finitely many $s\in [0,t]$. For the second term, observe that
\begin{align*}
\sum\limits_{0<s\leq t} \Vert \overline{\Lambda}_M(s)\Vert\leq \sum\limits_{j,k}\sum\limits_{0<s\leq t} |\Delta \overline{\Lambda}_{jk}(s)|<\infty,
\end{align*}
on account of the fact that $\overline{\Lambda}$ is a finite variation process. Thus, $G$ is a well-defined semimartingale.\par
For a matrix-valued semimartingale $Z$, let $\mathcal{E}(Z)$ denote the (matrix-valued) exponential of $Z$ and let $\mathbb{E}^R(Z)$ denote the (matrix-valued) right-stochastic exponential of $Z$ (cf. Chapter V in Protter (2005)). By applying Theorem V.10.63 of Protter (2005), we get
\begin{align*}
\mathcal{E}(F)(t)\mathcal{E}^R(\overline{\Lambda}_M^\top)(t)=\mathbb{I}
\end{align*}
for $F(t)=-\overline{\Lambda}_M^\top(t)+\sum_{0<s\leq t} (I+\Delta \overline{\Lambda}_M^\top(s))^{-1} (\Delta \overline{\Lambda}_M^{\top}(s))^2$. Because of
$\mathcal{E}^R(Z)=\mathcal{E}(Z^\top)^\top$ and $F^\top=-G$, the latter equation is equivalent to
\begin{align*}
\mathcal{E}(\overline{\Lambda}_M)(t)\mathcal{E}^R(-G)(t)=\mathbb{I},
\end{align*}
which proves the first equation of the assertion. In particular, we verified that
$\overline{q}(s,t)-\overline{q}(s,t-) = -(\Delta G(t))\overline{q}(s,t-)$, which implies that
\begin{align*}
  -(\Delta \overline{\Lambda}_M(t))\overline{q}(s,t)& =   -(\Delta \overline{\Lambda}_M(t))\overline{q}(s,t-) + (\Delta \overline{\Lambda}_M(t))(\Delta G(t))\overline{q}(s,t-) \\
  &=  -\Delta\overline{\Lambda}_M(t))(\mathbb{I}-\Delta G(t) )\overline{q}(s,t-) \\
   & =-(\Delta G(t)) \overline{q}(s,t-).
\end{align*}
Thus, the  second  equation of the assertion is also true.
\end{itemize}
\end{proof}
\end{lemma}
\begin{lemma}\label{itoinverse2}
Let $(\Phi',(\Lambda'_{jk})_{jk:j\neq k})$, $(\overline{\Phi},(\overline{\Lambda}_{jk})_{jk:j\neq k})$ be valuation bases.
\begin{itemize}
\item[a)] Let  $\d \kappa'(t)=\kappa'(t-)\d \Phi'(t)$ with $\kappa'(0)=1$ and $\d \overline{\kappa}(t)=\overline{\kappa}(t-)\d \overline{\Phi}(t)$ with $\kappa'(0)=1$.
      Then it holds that
\begin{align*}
\d \left(\frac{\kappa'(t)}{\overline{\kappa}(t)}\right)=\frac{\kappa'(t-)}{\overline{\kappa}(t-)}\Big( \d \Phi'- \d\widetilde{\overline{\Phi}}(t) -  \d [\Phi',\widetilde{\overline{\Phi}}](t) \Big) ,
\end{align*}
where $\widetilde{\overline{\Phi}}(t)=\overline{\Phi}(t)-[\overline{\Phi},\overline{\Phi}]^c(t)-\sum_{0<s\leq t} (1+\Delta \overline{\Phi}(s))^{-1} (\Delta \overline{\Phi}(s))^2$.
\item[b)] Let $p'(s,\d t)=p'(s,t-)\d \Lambda'_M(t)$ with $p'(s,s)=\mathbb{I}$ and $\overline{p}(s,\d t)=\overline{p}(s,t-)\d \overline{\Lambda}_M(t)$ with $\overline{p}(s,s)=\mathbb{I}$.  Suppose that $\overline{p}(s,t)$ is invertible with inverse $\overline{q}(s,t)$. Then it holds that
\begin{align*}
\d_t \left(p'(s,t)\overline{q}(s,t) \right)=p'(s,t-)\d (\Lambda'_M-\overline{\Lambda}_M)(t) \overline{q}(s,t).
\end{align*}
\end{itemize}
\end{lemma}
\begin{proof}
\begin{itemize}
\item[a)]  Integration by parts (Protter, 2005, Corollary II.6.2) and Lemma \ref{itoinverse}a) yield\begin{align*}
\d \left(\frac{\kappa'(t)}{\overline{\kappa}(t)}\right)&=\kappa'(t-)\d\left(\frac{1}{\overline{\kappa}(t)}\right)+\frac{1}{\overline{\kappa}(t-)}\d\kappa'(t)+\d \left[\frac{1}{\overline{\kappa}},\kappa'\right](t)\\
&=-\frac{\kappa'(t-)}{\overline{\kappa}(t-)}\d\widetilde{\overline{\Phi}}(t)+\frac{\kappa'(t-)}{\overline{\kappa}(t-)}\d \Phi'(t)-\frac{\kappa'(t-)}{\overline{\kappa}(t-)}\d [\widetilde{\overline{\Phi}},\Phi'](t) .
\end{align*}
\item[b)] Integration by parts (Protter, 2005, Corollary II.6.2) and Lemma \ref{itoinverse}b) yield
\begin{align*}
\d_t \left(p'(s,t)\overline{q}(s,t) \right)&=p'(s,t-)\overline{q}(s,\d t)+p'(s,\d t)\overline{q}(s,t-)+\d \left[p'(s,\cdot),\overline{q}(s,\cdot) \right](t) \\
&= - p'(s,t-)(\d \overline{\Lambda}_M(t)) \overline{q}(s,t)+p'(s,t-)(\d\Lambda'_M(t)) \overline{q}(s,t)\\
&=p'(s,t-) \d (\Lambda'_M-\overline{\Lambda}_M)(t)\overline{q}(s,t).
\end{align*}
\end{itemize}
\end{proof}

\end{document}